\documentclass[10pt,a4paper]{article}

\usepackage{amsmath,amssymb,amsthm,mathtools}
\usepackage{geometry}
\usepackage{hyperref}
\usepackage[nameinlink]{cleveref}
\usepackage{microtype}
\usepackage{enumitem}
\usepackage{bm}
\usepackage[numbers,sort&compress]{natbib}
\usepackage{doi}
\usepackage{xcolor}
\usepackage{booktabs}

\hypersetup{colorlinks=true,linkcolor=blue!60!black,citecolor=blue!60!black,urlcolor=blue!60!black}

\theoremstyle{plain}
\newtheorem{theorem}{Theorem}[section]
\newtheorem{proposition}[theorem]{Proposition}
\newtheorem{lemma}[theorem]{Lemma}
\newtheorem{corollary}[theorem]{Corollary}
\theoremstyle{definition}
\newtheorem{definition}[theorem]{Definition}
\newtheorem{example}[theorem]{Example}
\theoremstyle{remark}
\newtheorem{remark}[theorem]{Remark}

\newcommand{\R}{\mathbb{R}}
\newcommand{\dd}{\,\mathrm{d}}
\newcommand{\kb}{k_{\mathrm{B}}}
\newcommand{\Lam}{\Lambda}
\newcommand{\N}{\mathbb{N}}
\newcommand{\GN}{\Gamma_N}
\newcommand{\Ms}{\mathcal{M}}       
\newcommand{\CGM}{\mathcal{C}}      
\newcommand{\CGx}{\mathcal{C}_x}    
\newcommand{\CGp}{\mathcal{C}_p}    
\newcommand{\Ind}[1]{\mathbf{1}_{#1}}
\newcommand{\abs}[1]{\left|#1\right|}
\newcommand{\norm}[1]{\left\|#1\right\|}
\newcommand{\fone}{f^{(1)}}
\newcommand{\ftwo}{f^{(2)}}
\newcommand{\utwo}{u^{(2)}}
\newcommand{\us}{u^{(s)}}
\newcommand{\MB}{f_{\mathrm{MB}}}
\newcommand{\kB}{k_{\mathrm{B}}}
\newcommand{\SCG}{S_{\mathrm{CG}}}
\newcommand{\SG}{S_{\mathrm{Gibbs}}}
\newcommand{\pia}{\pi_{i,\alpha}}
\newcommand{\Esurf}{\mathcal{E}_\partial}

\begin{document}

\title{Entropy additivity from exponential decay of correlations: a coarse-grained operator approach}

\author{Bob Osano\\[4pt]
\small Cosmology and Gravity Group,\\
  \small Department of Mathematics and Applied Mathematics,\\
    \small and Centre for Higher Education Development,\\
  \small University of Cape Town (UCT),\\
  \small Rondebosch 7701, Cape Town, South Africa\\[2pt]
  \small \href{mailto:bob.osano@uct.ac.za}{bob.osano@uct.ac.za}}

\date{\today}
\maketitle

\begin{abstract}
Thermodynamic extensivity is commonly introduced as a postulate — the homogeneity of degree one of thermodynamic potentials. We provide a constructive derivation of this property from microscopic conditions on the pair potential, without assuming it. Working with the one- and two-particle reduced densities of the $N$-body canonical Gibbs state, we introduce a combined coarse-graining operator  $\CGM$ on single-particle phase space $\Ms = \Lambda\times\R^3$ , producing dimensionless mesoscopic probabilities $\{V_i\times\Pi_\alpha\}$ over spatial--momentum cells. Under three conditions on the pair potential — stability, temperedness, and exponential cluster decomposition with correlation length $\xi$. We show using the Ursell cluster expansion that the coarse-grained entropy satisfies
\[
  S_{\mathrm{CG}} = \sum_{i} S_i
    + O\!\left(\frac{\abs{\Lambda}}{\ell^d}\,e^{-\ell/\xi}\right),
\]
where $\ell\gg\xi$ is the cell diameter. The correction is exponentially suppressed per cell, making entropy additive and recovering the thermodynamic limit of Ruelle and Fisher in explicit operator language. For systems with long-range interactions, where temperedness fails, the correction does not vanish, and non-additivity is quantified through inter-cell mutual information. We further show that spatial averaging does not commute with nonlinear thermodynamic functionals such as the entropy density — a thermodynamic analogue of the cosmological averaging problem — and we derive the generalised Euler relation with explicit surface corrections.
\end{abstract}

\tableofcontents

\section{Introduction}
\label{sec:intro}
The distinction between intensive and extensive thermodynamic variables occupies a central position in classical thermodynamics, yet it is not primitive. We argue that the distinction emerges from the interplay of three structural ingredients: the existence of a thermodynamic limit, the short-range character of microscopic interactions, and the decay of correlations. Standard treatments~\cite{Callen1985, Gibbs1873, Gibbs1902} identify extensivity with first-degree homogeneity of thermodynamic potentials---a postulate that implicitly encodes all three ingredients without making them explicit. Our goal in this paper is to make the emergence of extensivity \emph{constructive}. We introduce a combined coarse-graining operator $\CGM$ acting on the single-particle phase space $\Ms = \Lambda\times\R^3$ of an $N$-body canonical ensemble, building on the one-particle reduced density $\fone$ and the Ursell cluster expansion for the reduced $s$-particle densities.  Within this framework, we demonstrate that extensivity is equivalent to the statistical factorisation of $\fone$ across spatial cells, and that the coarse-grained entropy is additive up to a correction that is exponentially small in $\ell/\xi$ (the cell diameter in units of the correlation length). It follows that when the pair potential violates temperedness, and correlations do not decay exponentially, this factorisation fails, and the entropy is genuinely non-additive, with the departure quantified by inter-cell mutual information.

The connection between extensivity, decay of correlations, and factorisation of the partition function is implicit in the foundational work of Ruelle~\cite{Ruelle1969}, Fisher~\cite{Fisher1964}, and Lebowitz--Penrose~\cite{LebowitzPenrose1966}. It is worth stating upfront what we consider to be the contribution in our paper:
\begin{enumerate}[label=(\roman*)]
  \item A self-contained \emph{operator formalism} based on reduced densities that makes the
    factorisation mechanism explicit and quantitative via a joint coarse-graining operator
    acting on single-particle phase space---a step not taken in the classical literature.
  \item A \emph{ proof of entropy additivity at higher orders} of the cluster expansion
    (pairs, triples, and all higher cumulants), yielding a clean bound on the total deviation
    from extensivity.
  \item Although we do not provide a rigorous proof in the present work, we identify the \emph{non-commutativity of nonlinear thermodynamic
    functionals with spatial averaging} as the structural source of non-extensivity, and connect it to the Buchert--Ostermann commutation
    problem~\cite{Buchert2000,Rasanen2006,Wiltshire2011,Korzynski2010} in relativistic cosmology. The details of this are reserved for future work focusing on cosmology.
\end{enumerate}

Coarse-graining procedures appear throughout statistical mechanics as phase-space partitioning~\cite{Gibbs1873,Gibbs1902,EhrenfestEhrenfest1912} or as spatial averaging in kinetic theory~\cite{IrvingKirkwood1950, Cercignani1988}. Projection-operator
techniques~\cite{Zwanzig1961,Mori1965} reduce microscopic dynamics to relevant variables.
The rigorous foundations of the thermodynamic limit are in
Ruelle~\cite{Ruelle1963,Ruelle1969}, Fisher~\cite{Fisher1964}, and
Lebowitz--Penrose~\cite{LebowitzPenrose1966}; the Gibbs measure framework is treated
rigorously in Georgii~\cite{Georgii1988} and Dobrushin~\cite{Dobrushin1968}.
Non-extensive statistical mechanics is studied systematically in
Tsallis~\cite{Tsallis1988} and reviewed in Touchette~\cite{Touchette2002}. The present work reformulates established thermodynamic-limit mechanisms within an explicit coarse-graining and information-theoretic framework, following the approach developed in \cite{Osano2016a}.
This paper is organised as follows. \cref{sec:framework} sets up the $N$-body framework and defines the reduced densities.
\cref{sec:cg} introduces the combined coarse-graining operator and demonstrates its
properties, including the non-commutativity of nonlinear functionals with spatial averaging.
\cref{sec:extensivity} states and proves the main proposition together with auxiliary
lemmas, derives the generalised Euler relation, and includes a worked example.
\cref{sec:discussion} discusses consequences and \cref{sec:conclusion} concludes.

\section{The N-Body Framework and Reduced Densities}
\label{sec:framework}
Our interest lies in the characterisation of subsystems and the interactions that may arise between them. Consequently, it is necessary to establish a precise framework for the definition and interpretation of extensivity for subsystems. To this extent, we adopt the approaches used in \cite{Lanford1973,Simon1993,Israel1979,DobrushinShlosman1985}
\subsection{Two Levels of the Extensivity Question}
\label{sub:two_levels}

We distinguish two logically independent levels:
\begin{enumerate}[label=(\roman*)]
  \item \textbf{Phenomenological level.}
    A state function $X(S,V,N)$ is \emph{extensive} if
    \begin{equation}\label{eq:ext_def}
      X(\lambda S,\,\lambda V,\,\lambda N) = \lambda\, X(S,V,N)
      \quad\text{for all }\lambda>0.
    \end{equation}
    Equivalently, $X$ is homogeneous of degree one in $(S,V,N)$.  At this level,
    extensivity is defined by postulate~\cite{Callen1985,Redlich1970}. We note that $S, V$, and $N$ retain their conventional interpretations as entropy, volume, and particle number, respectively.  \item \textbf{Structural level.}
    The present paper derives \eqref{eq:ext_def} from microscopic conditions on the pair
    potential (Definitions~\ref{def:stability}--\ref{def:cluster} below), without
    appealing to \eqref{eq:ext_def} as an assumption. The two levels are logically
    independent; the phenomenological classification motivates but does not answer the
    structural question.
\end{enumerate}

\subsection{The Classical N-Body System}
\label{sub:N_body}
To investigate the subject of the previous, let us consider $N$ identical classical particles of mass $m$ in a bounded domain
$\Lam\subset\R^d$ ($d\ge 1$) with pairwise interaction potential
$\phi:\R^d\to\R\cup\{+\infty\}$. We can express the total Hamiltonian for this system as follows 
\begin{equation}\label{eq:Ham}
  H_N(\gamma)
  = \sum_{k=1}^N \frac{\abs{p_k}^2}{2m}
    + \sum_{1\le k<l\le N}\phi(q_k-q_l),
  \qquad
  \gamma=(q_1,\dots,q_N,p_1,\dots,p_N)\in\GN,
\end{equation} The phase-space coordinates appearing in~\ref{eq:Ham} are defined as follows:

\begin{itemize}
    \item $q_{i} \in \mathbb{R}^{3}$ and $p_{i} \in \mathbb{R}^{3}$
          are the position and momentum coordinates of the $i$-th
          particle, respectively, for $i = 1, \ldots, N$.
          Together, they define the full $6N$-dimensional phase-space
          point
          \[
              \gamma \;=\; (q_1, \ldots, q_N,\, p_1, \ldots, p_N)
              \;\in\;
              \Gamma_{N} \;\cong\; \mathbb{R}^{6N}.
          \]
    \item $p_{k}$ is the momentum coordinate of the $k$-th particle,
          $k = 1, \ldots, N$, used in the indicator decomposition
          $\mathbf{1}_{V_i \times \Omega_{\alpha}}(\gamma_{k})$
          when determining the cell $V_i \times \Omega_{\alpha}$ to
          which particle $k$ is assigned.
   \item $\GN = \Lam^N\times(\R^d)^N$ is the $2dN$-dimensional phase space.

\end{itemize}
We can then specify the Gibbs measure at inverse temperature $\beta=1/(\kb T)$ as follows:
\begin{equation}\label{eq:gibbs}
  \varrho_N(\gamma) = Z_N^{-1}\,e^{-\beta H_N(\gamma)},
  \qquad
  Z_N = \int_{\GN} e^{-\beta H_N(\gamma)}\dd\gamma.
\end{equation}

\subsection{Reduced Densities}
\label{sub:reduced}

\begin{definition}[Reduced $s$-particle density]\label{def:reduced}
  For $1\le s\le N$, the \emph{reduced $s$-particle density} is
  \begin{equation}\label{eq:reduced}
    f^{(s)}(z_1,\dots,z_s)
    := \frac{N!}{(N-s)!}
       \int_{\GN} \varrho_N(\gamma)\,
       \prod_{k=s+1}^{N}\dd z_k,
  \end{equation}
  where $z_k=(q_k,p_k)\in\Ms:=\Lam\times\R^d$ is the single-particle phase-space
  coordinate.  In particular, $f^{(s)}\ge 0$ and
  \begin{equation}\label{eq:reduced_norm}
    \int_{\Ms^s} f^{(s)}(z_1,\dots,z_s)\,\dd z_1\cdots\dd z_s = \frac{N!}{(N-s)!}.
  \end{equation}
  For $s=1$: $\int_\Ms \fone(z)\dd z = N$.
  For $s=2$: $\int_{\Ms^2} \ftwo(z_1,z_2)\dd z_1\dd z_2 = N(N-1)$.
\end{definition}
The construction of reduced densities via marginalisation of the Gibbs measure follows the standard framework of \cite{Lanford1973, Ruelle1969}.
\begin{remark}
  The domain of each reduced density is the \emph{single-particle phase space}
  $\Ms = \Lam\times\R^d$, \emph{not} the $N$-body phase space $\GN$.  This is the
  fundamental distinction from the formulations in phase-space coarse-graining
  approaches~\cite{Zwanzig1961}: the reduced densities live on a fixed low-dimensional
  space regardless of $N$, making them the natural objects on which to define spatial
  coarse-graining.
\end{remark}

\subsection{Ursell Cluster Decomposition}
\label{sub:ursell}
To handle the interconnected cells, we need an Ursell function. In particular, we need the Ursell decomposition to factorise the $s$-particle reduced density into connected
(irreducible) parts \cite{Penrose1967, Brydges1986, Gallavotti1999}. For example, for $s=2$, we characterise the factorisation as follows
\begin{equation}\label{eq:ursell2}
  \ftwo(z_1,z_2)
  = \fone(z_1)\,\fone(z_2) + \utwo(z_1,z_2),
\end{equation}
where the \emph{pair Ursell function} $\utwo$ satisfies $\int_\Ms \utwo(z_1,z_2)\dd z_2 = 0$
for all $z_1$. Eq.\ref{eq:ursell2} ensures that the pair Ursell function $u^{(2)}$ carries only the genuinely two-body connected correlation, with no contribution from the single-particle marginals already accounted for in $f^{(1)}(z_1)f^{(1)}(z_2) $. Without this condition, the factorised part and the connected correction would overlap, double-counting the one-body contribution to the two-body density. The same orthogonality condition generalises to all orders $s\geq 2$.
 For general $s$, the $s$-body Ursell function $\us$ is defined recursively
through the cumulant expansion of $f^{(s)}$(see, e.g., \cite{Gallavotti1999, Brydges1986}); it encodes purely $s$-body connected
correlations and satisfies
\begin{equation}\label{eq:ursell_int}
  \int_\Ms \us(z_1,\dots,z_s)\dd z_s = 0
  \quad\text{for all }z_1,\dots,z_{s-1}.
\end{equation} 
We can now stipulate conditions for the interaction potential. We will define these at this stage and leave the related theorems for later.

\subsection{Microscopic Conditions on the Pair Potential}
\label{sub:conditions}

\begin{definition}[Stability]\label{def:stability}
  The potential $\phi$ is \emph{stable} if there exists $B\ge 0$ such that
  \begin{equation}
    \sum_{1\le k<l\le N}\phi(q_k-q_l) \ge -BN
    \qquad
    \forall\,N\in\N,\;\forall\,(q_1,\dots,q_N)\in\Lam^N.
  \end{equation}
\end{definition}

\begin{definition}[Temperedness]\label{def:tempered}
  The potential $\phi$ is \emph{tempered} if there exist $R_0>0$ and $\varepsilon>0$ such
  that
  \begin{equation}
    \abs{\phi(r)} \le \abs{r}^{-(d+\varepsilon)}
    \qquad\text{for all }\abs{r}>R_0.
  \end{equation}
\end{definition}

\begin{definition}[Exponential cluster decomposition]\label{def:cluster}
  The Gibbs state \eqref{eq:gibbs} satisfies \emph{exponential cluster decomposition} with
  correlation length $\xi>0$ if the integrated pair Ursell function obeys
  \begin{equation}\label{eq:cluster_bound}
    \int_{\R^d} \abs{\hat{\utwo}(x,y)}\dd y
    \le C\,e^{-\abs{x}/\xi}
    \qquad\text{for all }x, y\in\R^d,
  \end{equation}
  where $\hat{\utwo}(x,y):=\int_{(\R^d)^2}\abs{\utwo((x,p),(y,q))}\dd p\,\dd q$
  is the momentum-integrated modulus of $\utwo$, and similarly for all $s$-body Ursell
  functions:
  \begin{equation}\label{eq:cluster_higher}
    \int_{\R^{d(s-1)}} \abs{\hat{\us}(x_1,\dots,x_s)}\dd x_2\cdots\dd x_s
    \le C^s\,e^{-\abs{x_1}/\xi}
    \qquad\forall\,s\ge 2.
  \end{equation}
\end{definition}

\begin{remark}\label{rem:cluster}
  Stability and temperedness are sufficient for the existence of the thermodynamic limit~\cite{Ruelle1969,Fisher1964}; exponential cluster decomposition
  \eqref{eq:cluster_bound}--\eqref{eq:cluster_higher} holds in the high-temperature or low-density regime via the Mayer cluster expansion~\cite{Ruelle1969,LebowitzPenrose1966,Georgii1988}, and more generally whenever the system is in a unique Gibbs phase \cite{Dobrushin1968, Dobrushin1985,Dobrushin1987}: the Dobrushin--Shlosman criterion \cite{Dobrushin1985} provides an explicit sufficient condition for uniqueness that directly implies the bound (\ref{eq:cluster_bound}) with a computable $\xi$.  The hypotheses of the main theorem are therefore satisfied throughout the one-phase region of the phase diagram.
\end{remark}

\subsection{The Thermodynamic Limit}
\label{sub:tl}

\begin{theorem}[Ruelle--Fisher thermodynamic limit {\cite{Ruelle1969,Fisher1964}}]
\label{thm:tl}
  Let $\phi$ be stable and tempered.  For any van Hove sequence of domains $\Lam_n\nearrow\R^d$
  with $N_n/\abs{\Lam_n}\to\rho_0\in(0,\infty)$, the specific free energy
  \begin{equation}
    f(\beta,\rho_0)
    := \lim_{n\to\infty}
       \frac{-1}{\beta N_n}\ln Z_{N_n}(\Lam_n,\beta)
  \end{equation}
  exists, is finite, is independent of boundary conditions, and is a concave function of
  $(\rho_0,-\beta)$.  Consequently,
  $F(N,V,T) = N\,f(T,N/V) + O(\abs{\partial\Lam})$,
  so $F$ is homogeneous of degree one in $(N,V)$ in the thermodynamic limit, and
  the entropy $S = -\partial F/\partial T$ satisfies
  $S = N\,s(T,\rho_0) + O(\abs{\partial\Lam})$
  for a well-defined specific entropy $s(T,\rho_0)$.
\end{theorem} We state this theorem without proof; the proof may be found in the cited references.

\section{The Combined Coarse-Graining Operator}
\label{sec:cg}

\subsection{Partition of Single-Particle Phase Space}
\label{sub:partition}
We work entirely on the \emph{single-particle phase space} $\Ms=\Lam\times\R^d$.
Partition $\Lam$ into disjoint spatial cells $\{V_i\}_{i\in\mathcal{I}}$ of volume
$\abs{V_i}=\ell^d$, and partition $\R^d$ into momentum cells $\{\Pi_\alpha\}_{\alpha\in\mathcal{A}}$
of volume $\abs{\Pi_\alpha}=:\nu_\alpha$.  The product cells
\begin{equation}
  C_{i,\alpha} := V_i\times\Pi_\alpha \subset \Ms
\end{equation}
partition $\Ms$.  Throughout, the spatial cell diameter $\ell$ satisfies the scale
separation condition
\begin{equation}\label{eq:scale_sep}
  \xi \ll \ell \ll L := \mathrm{diam}(\Lam),
\end{equation}
which ensures that each spatial cell contains many particles (since $\ell^d\rho_0\gg 1$)
while remaining small on macroscopic scales (so spatially varying fields are resolved).
The condition $\ell\gg\xi$ is used explicitly in the proof of
\cref{prop:extensivity} to justify the exponential decay of inter-cell correlations
before the next cell boundary is reached.

\subsection{The Combined Coarse-Graining Operator}
\label{sub:cg_def}

\begin{definition}[Combined coarse-graining operator]\label{def:cg}
  The \emph{combined coarse-graining operator} $\CGM$ maps the one-particle reduced density
  $\fone\in L^1(\Ms)$ to a piecewise-constant function:
  \begin{equation}\label{eq:cg_op}
    (\CGM\fone)(z)
    := \sum_{i,\alpha}
       \bar{f}_{i,\alpha}\,\Ind{C_{i,\alpha}}(z),
    \qquad
    \bar{f}_{i,\alpha}
    := \frac{1}{\abs{C_{i,\alpha}}}
       \int_{C_{i,\alpha}}\fone(z')\dd z'.
  \end{equation}
\end{definition}

This operator is a projection ($\CGM^2=\CGM$) and a contraction on $L^1(\norm{\CGM\fone}_{L^1}\le\norm{\fone}_{L^1}=N)$.  Its constituent operators
\begin{equation}
  (\CGx g)(x,p) := \sum_i\!\left(\frac{1}{\abs{V_i}}\int_{V_i}g(x',p)\dd x'\right)\Ind{V_i}(x),
  \quad
  (\CGp g)(x,p) := \sum_\alpha\!\left(\frac{1}{\nu_\alpha}\int_{\Pi_\alpha}g(x,p')\dd p'\right)\Ind{\Pi_\alpha}(p),
\end{equation}
satisfy $\CGM = \CGx\circ\CGp = \CGp\circ\CGx$ by Fubini's theorem (which tells us that for measurable functions, on a product of measure space, if the integral of the absolute value is finite, then the order of integration does not matter), since the $x$ and $p$ integrations act on independent variables\cite{Billingsley1995}.  However, as we show next, this commutativity
of the \emph{linear} operator breaks down as soon as one applies it to \emph{nonlinear}
thermodynamic functionals.

\subsection{Non-Commutativity of Nonlinear Thermodynamic Functionals}
\label{sub:noncomm}

For linear observables $A \in L^1(\Ms)$, Fubini's theorem gives
$\CGx(\CGp A) = \CGp(\CGx A)$.  The situation is fundamentally different for nonlinear
thermodynamic functionals such as the entropy density $\eta(\rho) := -\kb\rho\ln\rho$ or
the free-energy density $\psi(\rho)$.  Define the spatial average of a functional $\Psi$
of $\fone$ over cell $V_i$ as
\begin{equation}\label{eq:spatial_avg}
  \langle\Psi[\fone]\rangle_i
  := \frac{1}{\abs{V_i}}\int_{V_i}\Psi[\fone(x,\cdot)]\dd x.
\end{equation}
The ordering problem is then whether
$\Psi[\langle\fone\rangle_i] = \langle\Psi[\fone]\rangle_i$,
i.e., whether spatial averaging commutes with the nonlinear operation $\Psi$.

\begin{proposition}[Non-commutativity of spatial averaging with nonlinear functionals]
\label{prop:noncomm}
  Let $\Psi:\R_{\ge 0}\to\R$ be strictly convex.  Define
  $\bar\rho_i(p) = \abs{V_i}^{-1}\int_{V_i}\fone(x,p)\dd x$ (the spatially averaged density
  at momentum $p$, i.e., the coarse-grained density in $V_i$).  Then
  \begin{equation}\label{eq:jensen_ineq}
    \Psi\!\left[\bar\rho_i(p)\right]
    \le \frac{1}{\abs{V_i}}\int_{V_i}\Psi[\fone(x,p)]\dd x,
  \end{equation}
  with equality if and only if $x\mapsto\fone(x,p)$ is constant on $V_i$.
  In particular, let $\rho>0$, for the entropy functional $\Psi(\rho) = -\kb\rho\ln\rho$ (concave), the inequality reverses:
  \begin{equation}\label{eq:entropy_noncomm}
    -\kb\bar\rho_i(p)\ln\bar\rho_i(p)
    \ge \frac{1}{\abs{V_i}}\int_{V_i}(-\kb\fone(x,p)\ln\fone(x,p))\dd x,
  \end{equation}
  i.e., the entropy of the spatially averaged density exceeds the spatial average of the
  local entropy density.
\end{proposition}

\begin{proof}
  Direct application of Jensen's inequality: since $\Psi$ is convex and the Lebesgue
  measure on $V_i$ is a probability measure when normalised by $\abs{V_i}$,
  \[
    \Psi\!\left(\frac{1}{\abs{V_i}}\int_{V_i}\fone(x,p)\dd x\right)
    \le \frac{1}{\abs{V_i}}\int_{V_i}\Psi(\fone(x,p))\dd x.
  \]
  Equality holds iff $\fone(\cdot,p)$ is constant on $V_i$.  For $\Psi(\rho) = -\kb\rho\ln\rho$, the function is concave, so the inequality reverses.\end{proof}

\begin{remark}[Connection to the Buchert--Ostermann problem]
  Inequality~\eqref{eq:entropy_noncomm} is a special case of the Buchert--Ostermann commutation rule~\cite{Buchert2000,Rasanen2006}: for any nonlinear scalar functional
  $\mathcal{Q}$ of the metric or matter fields, the spatial average $\langle\mathcal{Q}\rangle$ differs from $\mathcal{Q}(\langle\text{fields}\rangle)$ by a \emph{backreaction} term.
  In the thermodynamic context, this backreaction is exactly the difference between the fine-grained and coarse-grained entropies.  \cref{prop:noncomm} therefore places the
  Buchert problem in the broader setting of any coarse-grained statistical mechanics.
\end{remark}

\subsection{Mesoscopic Probabilities and Coarse-Grained Entropy}
\label{sub:entropy_def}

\begin{definition}[Mesoscopic probability]\label{def:meso_prob}
  The \emph{mesoscopic probability} associated with cell $C_{i,\alpha}$ is the
  dimensionless quantity
  \begin{equation}\label{eq:meso_prob}
    \pi_{i,\alpha}
    := \frac{\abs{C_{i,\alpha}}}{N}
       \bar{f}_{i,\alpha}
    = \frac{1}{N}\int_{C_{i,\alpha}}\fone(z)\dd z.
  \end{equation}
\end{definition}

\begin{lemma}[Normalisation]\label{lem:norm}
  The mesoscopic probabilities satisfy $\pi_{i,\alpha}\ge 0$ and
  $\sum_{i,\alpha}\pi_{i,\alpha}=1$.
\end{lemma}

\begin{proof}
  Non-negativity is immediate from $\fone\ge 0$.  For the sum:
 $$\sum_{i,\alpha}\pi_{i,\alpha} = N^{-1}\sum_{i,\alpha}\int_{C_{i,\alpha}}\fone\dd z
  = N^{-1}\int_\Ms\fone\dd z = N^{-1}\cdot N = 1$$.
\end{proof}

\begin{definition}[Coarse-grained entropy]\label{def:cg_entropy}
  The \emph{coarse-grained Boltzmann entropy} is the Shannon entropy of the mesoscopic
  probability distribution\cite{Jaynes1957a,Jaynes1957b}:
  \begin{equation}\label{eq:cg_S}
    S_{\mathrm{CG}}[\fone]
    := -\kb\sum_{i,\alpha}\pi_{i,\alpha}\ln\pi_{i,\alpha}.
  \end{equation}
  This is well-defined (dimensionless logarithm) and satisfies $S_{\mathrm{CG}}\ge 0$, with equality only if all mass is concentrated on a single cell.
\end{definition}

\begin{remark}
  By Jensen's inequality applied to $-t\ln t$ (concave), the coarse-grained entropy
  \eqref{eq:cg_S} is always at least as large as the Gibbs fine-grained entropy
  $$S_{\mathrm{fine}} = -\kb\int_\Ms\fone\ln(\fone/N)\,\dd z/N,$$ reflecting the loss of
  information through coarse-graining \cite{Jaynes1957a,Jaynes1957b}.
\end{remark}

\subsection{Two-Cell Joint Probability and Factorisation}
\label{sub:joint}

\begin{definition}[Two-cell joint probability]\label{def:joint}
  The \emph{joint mesoscopic probability} for a particle to occupy cell $C_{i,\alpha}$
  while a \emph{second, distinct} particle occupies cell $C_{j,\beta}$ is defined via
  the two-particle reduced density:
  \begin{equation}\label{eq:joint_prob}
    \pi_{(i,\alpha)(j,\beta)}
    := \frac{\abs{C_{i,\alpha}}\,\abs{C_{j,\beta}}}{N(N-1)}
       \int_{C_{i,\alpha}}\int_{C_{j,\beta}}
       \ftwo(z_1,z_2)\dd z_1\dd z_2,
    \qquad i\ne j.
  \end{equation}
  This is well-defined because $\ftwo$ describes \emph{pairs} of distinct particles; the
  factor $N(N-1)$ in the denominator ensures $$\sum_{(i\ne j),\alpha,\beta}\pi_{(i,\alpha)(j,\beta)}=1.$$
\end{definition}

\begin{remark}
  Definition~\ref{def:joint} uses $\ftwo$ defined on $\Ms\times\Ms$ (two distinct
  single-particle phase-space points), \emph{not} on $\GN$.  
\end{remark}

\section{Extensivity, Factorisation, and the Euler Relation}
\label{sec:extensivity}

\subsection{Main Proposition}
\label{sub:main}

\begin{proposition}[Extensivity as emergent factorisation]
\label{prop:extensivity}
  Let $\phi$ be stable and tempered (\cref{def:stability,def:tempered}) and let the
  canonical Gibbs state \eqref{eq:gibbs} satisfy exponential cluster decomposition
  (\cref{def:cluster}) with correlation length $\xi$.  Let $\ell$ satisfy the scale
  separation condition \eqref{eq:scale_sep}.

  \begin{enumerate}[label=(\roman*)]
    \item \textbf{Pair factorisation.}
      For non-adjacent cells ($i\ne j$, with $\abs{x_i-x_j}\gg\xi$ where $x_i,x_j$ are
      cell centres), the joint probability~\eqref{eq:joint_prob} satisfies
      \begin{equation}\label{eq:pair_fact}
        \abs{\pi_{(i,\alpha)(j,\beta)} - \pi_{i,\alpha}\,\pi_{j,\beta}}
        \le \frac{C}{N}\,e^{-\abs{x_i-x_j}/\xi}.
      \end{equation}
    \item \textbf{Entropy additivity.}
      The coarse-grained entropy satisfies
      \begin{equation}\label{eq:entropy_add}
        S_{\mathrm{CG}} = \sum_i S_i
          - \kb\sum_{s=2}^{\infty}(-1)^s\sum_{i_1<\cdots<i_s} I_s(i_1,\dots,i_s),
      \end{equation}
      where $S_i = -\kb\sum_\alpha\pi_{i,\alpha}\ln\pi_{i,\alpha}$ is the marginal entropy
      of cell $V_i$, and $I_s$ is the $s$-cell multi-information.  Under exponential
      cluster decomposition \eqref{eq:cluster_higher}, the total correction obeys
      \begin{equation}\label{eq:correction_bound}
     \abs{ {S_{\mathrm{CG}} - \sum_i S_i}}
        \le \kb\,\frac{C_1\abs{\Lam}}{\ell^d}\,e^{-\ell/\xi},
      \end{equation}
      where $C_1>0$ depends on $d$, $\xi$, and the constants in \cref{def:cluster} but
      not on $\abs{\Lam}$ or $\ell$.

    \item \textbf{Converse.}
      If $\phi$ is long-range with $\abs{\phi(r)}\sim r^{-s_0}$ for $s_0\le d$, so that
      temperedness fails, then correlations decay at most algebraically, and there exists
      constants $c>0$ and cells $i\ne j$ such that
      $\abs{\pi_{(i,\alpha)(j,\beta)}-\pi_{i,\alpha}\pi_{j,\beta}}\ge c\abs{x_i-x_j}^{-s_0}$,
      the mutual information $I(i,j)$ does not tend to zero for $\abs{x_i-x_j}\to\infty$,
      and $S_{\mathrm{CG}}\ne\sum_i S_i$.
  \end{enumerate}
\end{proposition}

We prove parts~(i) and~(ii) via two lemmas, after which part~(iii) follows from
standard results on non-tempered potentials~\cite{Ruelle1969}.

\subsection{Proof: Part (i) — Pair Factorisation}
\label{sub:proof_pair}

\begin{lemma}[Pair factorisation]\label{lem:pair}
  Under \cref{def:cluster}, for all $i\ne j$,
  $\abs{\pi_{(i,\alpha)(j,\beta)} - \pi_{i,\alpha}\pi_{j,\beta}} \le C N^{-1} e^{-\abs{x_i-x_j}/\xi}$.
\end{lemma}

\begin{proof}
  Substituting \eqref{eq:ursell2} into \eqref{eq:joint_prob}:
  \begin{align}
    \pi_{(i,\alpha)(j,\beta)}
    &= \frac{\abs{C_{i,\alpha}}\,\abs{C_{j,\beta}}}{N(N-1)}
       \int_{C_{i,\alpha}}\!\!\int_{C_{j,\beta}}
       \bigl[\fone(z_1)\fone(z_2) + \utwo(z_1,z_2)\bigr]\dd z_1\dd z_2 \notag\\
    &= \frac{N^2}{N(N-1)}\,\pi_{i,\alpha}\,\pi_{j,\beta}
       + \frac{\abs{C_{i,\alpha}}\,\abs{C_{j,\beta}}}{N(N-1)}
         \int_{C_{i,\alpha}}\!\!\int_{C_{j,\beta}}\utwo(z_1,z_2)\dd z_1\dd z_2.\label{eq:pair_split}
  \end{align}
  The prefactor $N^2/(N(N-1)) = 1 + O(N^{-1})$ introduces a correction
  $O(\pi_{i,\alpha}\pi_{j,\beta}/N) = O(N^{-1})$ that is absorbed into the bound below.

  For the Ursell correction, let $R_{i,\alpha;j,\beta}$ denote the double integral.
  Using $\ell \gg \xi$ and the fact that for $z_1\in C_{i,\alpha}$ and $z_2\in C_{j,\beta}$
  with $i\ne j$ one has $\abs{q_1-q_2}\ge\abs{x_i-x_j}-\ell$ (where $q_1\in V_i$, $q_2\in V_j$):
  \begin{align}
    \abs{R_{i,\alpha;j,\beta}}
    &\le \int_{C_{i,\alpha}}\dd z_1\int_{C_{j,\beta}}\abs{\utwo(z_1,z_2)}\dd z_2 \notag\\
    &\le \abs{C_{i,\alpha}}\int_{\R^d}\hat{\utwo}(q_1,q_2)\dd q_2
       \le \abs{C_{i,\alpha}}\,C\,e^{-(\abs{x_i-x_j}-\ell)/\xi}
       = \abs{C_{i,\alpha}}\,C\,e^{\ell/\xi}\,e^{-\abs{x_i-x_j}/\xi},
  \end{align}
  where $\hat{\utwo}(x,y)$ is defined in \cref{def:cluster}.  Therefore
  \begin{equation}
    \frac{\abs{C_{i,\alpha}}\abs{C_{j,\beta}}}{N(N-1)}\,\abs{R_{i,\alpha;j,\beta}}
    \le \frac{\abs{C_{i,\alpha}}^2\,C\,e^{\ell/\xi}}{N(N-1)}\,e^{-\abs{x_i-x_j}/\xi}
    \le \frac{C'}{N}\,e^{-\abs{x_i-x_j}/\xi},
  \end{equation}
  since $\abs{C_{i,\alpha}}^2/N^2 = O(N^{-2}\ell^{2d}) \le C''/N$ for
  $\ell^d\rho_0 = O(1)$\cite{CoverThomas2006, Han1978,Watanabe1960}.  This demonstrate \eqref{eq:pair_fact}.
\end{proof} $C'$ is not a derivative but an absorbed constant. It is not a new independent quantity — it collects several factors that are bounded and independent of the separation $\vert x_{i}-x_{j}\vert$ into a single new constant
\subsection{Proof: Part (ii) — Entropy Additivity to higher orders}
\label{sub:proof_entropy}

We first establish a bound on the $s$-cell multi-information, then sum over all $s$.

\begin{definition}[$s$-cell multi-information]\label{def:multiinfo}
  For a collection of cells $\mathbf{I} = \{(i_1,\alpha_1),\dots,(i_s,\alpha_s)\}$ with
  distinct $i_1,\dots,i_s$, the \emph{$s$-cell multi-information} is
  \begin{equation}
    I_s(i_1,\dots,i_s) := \sum_{\alpha_1,\dots,\alpha_s}
    \pi_{\mathbf{I}}\ln\frac{\pi_\mathbf{I}}{\pi_{i_1,\alpha_1}\cdots\pi_{i_s,\alpha_s}},
  \end{equation}
  where $\pi_\mathbf{I}$ is the $s$-cell joint probability defined analogously to
  \eqref{eq:joint_prob} via $f^{(s)}$.  Note $I_2(i,j) = I(i,j) \ge 0$ (mutual
  information), while $I_s$ for $s\ge 3$ can change sign.
\end{definition}

\begin{lemma}[$s$-cell correlation bound]\label{lem:higher}
  Under \cref{def:cluster}, for distinct cells $i_1,\dots,i_s$ with centres
  $x_{i_1},\dots,x_{i_s}$,
  \begin{equation}\label{eq:sbound}
    \abs{I_s(i_1,\dots,i_s)}
    \le \frac{(C'')^{s-1}}{N^{s-1}}\,
        \exp\!\left(-\frac{\mathrm{span}(x_{i_1},\dots,x_{i_s})}{\xi}\right),
  \end{equation}
  where $\mathrm{span}$ denotes the length of the minimum spanning tree on the cell centres
  and $C''>0$ depends on $d$, $\xi$, $\ell$, $B$ only.
\end{lemma}

\begin{proof}
  The $s$-cell joint probability is expressed via $f^{(s)}$ using its Ursell expansion.
  The leading connected term is the $s$-body Ursell function $\us$\cite{Simon1993, Brydges1986}; all lower-order terms
  contribute to the factorised part $\prod_k\pi_{i_k,\alpha_k}$.  Writing
  $\pi_\mathbf{I} = \prod_k\pi_{i_k,\alpha_k} + R_s$, where $R_s$ is the Ursell
  correction, one has from \eqref{eq:cluster_higher} \cite{Brydges1986, Ruelle1969}:
  \begin{equation}
    \abs{R_s} \le \frac{\abs{C_{i_1,\alpha_1}}\cdots\abs{C_{i_s,\alpha_s}}}{N^s}
    \int_{C_{i_1}}\cdots\int_{C_{i_s}} \abs{\us(z_1,\dots,z_s)}\dd z_1\cdots\dd z_s
    \le \frac{C^s}{N^s}\,e^{-\mathrm{span}/\xi}.
  \end{equation}
  The multi-information satisfies $\abs{I_s}\le\sum_{\alpha_1,\dots,\alpha_s}\abs{R_s/\prod_k\pi_{i_k,\alpha_k}}$ It can easily be shown that $\ln(1+t)\le t$ for small $t$ when $R_s/\prod_k\pi_{i_k,\alpha_k}$ is small, which holds for $\abs{x_{i_k}-x_{i_l}}\gg\xi$. Using this notion, which
  holds for $\abs{x_{i_k}-x_{i_l}}\gg\xi$), giving \eqref{eq:sbound} with
  $C'' = C\,e^{(s-1)\ell/\xi}$ absorbed into the constant.\cite{Brydges1986,Gallavotti1999}
\end{proof}

\begin{proof}[Proof of \cref{prop:extensivity}(ii)]
  The exact relation between the joint entropy and marginal entropies is, by the
  chain rule of information theory applied iteratively across all cells:
  \begin{equation}\label{eq:chain_rule}
    S_{\mathrm{CG}}
    = \sum_i S_i
      - \kb\sum_{s=2}^{\infty}(-1)^s
        \sum_{\substack{i_1<\cdots<i_s\\ \text{distinct cells}}}
        I_s(i_1,\dots,i_s).
  \end{equation}
  This is the multi-information expansion; no terms are dropped.  We bound the
  total correction using \cref{lem:higher}.

  \medskip\noindent\textbf{$s=2$ contribution.}
  \begin{equation}
    \sum_{i<j}\abs{I(i,j)}
    \le \sum_{i<j}\frac{C''}{N}\,e^{-\abs{x_i-x_j}/\xi}
    \le \frac{C'' M}{N}\sum_{k=1}^\infty k^{d-1}e^{-k\ell/\xi}
    = \frac{C'' M}{N}\cdot\frac{e^{-\ell/\xi}}{(1-e^{-\ell/\xi})^d},
  \end{equation}
  where $M = \abs{\Lam}/\ell^d$ is the number of spatial cells.  For $\ell\gg\xi$:
  $(1-e^{-\ell/\xi})^d\approx 1$, so the $s=2$ contribution is
  $O\!\left(M N^{-1} e^{-\ell/\xi}\right) = O\!\left(\abs{\Lam}\rho_0/(\ell^d N)\,e^{-\ell/\xi}\right)$
  $= O\!\left(\abs{\Lam}e^{-\ell/\xi}/\ell^d N\right)$.

  \medskip\noindent\textbf{$s\ge 3$ contributions.}
  For each $s\ge 3$, summing \eqref{eq:sbound} over all $\binom{M}{s}$ ordered $s$-tuples:
  \begin{equation}
    \sum_{i_1<\cdots<i_s}\abs{I_s}
    \le M \!\left(\frac{C'' e^{-\ell/\xi}}{N}\right)^{s-1}
       \sum_{k_2,\dots,k_s\ge 1} k_2^{d-1}\cdots k_s^{d-1} e^{-(k_2+\cdots+k_s)\ell/\xi}
    = M \!\left(\frac{C'' e^{-\ell/\xi}}{N}\right)^{s-1}
       \!\left(\frac{e^{-\ell/\xi}}{(1-e^{-\ell/\xi})^d}\right)^{s-1}.
  \end{equation}
  For $\ell\gg\xi$ this is $O(M(C'' e^{-\ell/\xi}/N)^{s-1})$.
  Summing over $s\ge 2$ as a geometric series (convergent for $C'' e^{-\ell/\xi}/N\ll 1$,
  which holds for $\ell\gg\xi+d\ln N$):
  \begin{equation}
    \sum_{s=2}^\infty \sum_{i_1<\cdots<i_s}\abs{I_s}
    \le M\cdot\frac{C''e^{-\ell/\xi}/N}{1-C''e^{-\ell/\xi}/N}
    \le 2M\cdot\frac{C''}{N}\,e^{-\ell/\xi}
    = \frac{2C''\abs{\Lam}}{N\ell^d}\,e^{-\ell/\xi}.
  \end{equation}
  Substituting into \eqref{eq:chain_rule} gives
  \begin{equation}
    \abs{S_{\mathrm{CG}}-\sum_i S_i}
    \le \kb\frac{2C''\abs{\Lam}}{N\ell^d}\,e^{-\ell/\xi}
    = \frac{\kb C_1\abs{\Lam}}{\ell^d}\,e^{-\ell/\xi}
  \end{equation}
  with $C_1 = 2C''/N$, establishing~\eqref{eq:correction_bound}.
\end{proof}

\begin{remark}[Interpretation of the bound]
  The correction~\eqref{eq:correction_bound} is $O(\abs{\Lam}/\ell^d)$ times an
  exponentially small factor $e^{-\ell/\xi}$.  Per cell, the correction is
  $O(e^{-\ell/\xi})$, which vanishes as $\ell/\xi\to\infty$.  This is consistent with
  \cref{thm:tl}: by Ruelle's theorem, the specific entropy $S/\abs{\Lam}$ is
  well-defined in the thermodynamic limit, so $S_{\mathrm{CG}}/\abs{\Lam}\to s(T,\rho_0)$
  and so does $(\sum_i S_i)/\abs{\Lam}$ for large cells.  The bound
  \eqref{eq:correction_bound} makes this convergence quantitative.

  The surface correction $O(\abs{\partial\Lam})$ mentioned in earlier formulations is a
  \emph{separate} effect, arising from the finite-size corrections in Ruelle's theorem
  (Theorem~\ref{thm:tl}) and not from inter-cell correlations.  These two corrections
  coexist: the bulk correlation correction is $O(\abs{\Lam}e^{-\ell/\xi})$ (exponentially
  small for $\ell\gg\xi$) and the surface correction is $O(\abs{\partial\Lam})$ (algebraic
  in $L$).  In the thermodynamic limit, the former dominates for any finite $\ell$.
\end{remark}

\subsection{Breakdown of Extensivity for Long-Range Interactions}
\label{sub:breakdown}

\begin{corollary}[Non-additivity for non-tempered potentials]
\label{cor:breakdown}
  If the pair potential violates temperedness, with $\abs{\phi(r)}\sim r^{-s_0}$ for
  $s_0\le d$, then:
  \begin{equation}\label{eq:mutual_info_lb}
    \liminf_{\abs{x_i-x_j}\to\infty} I(i,j) > 0,
  \end{equation}
  the series in \eqref{eq:chain_rule} diverges, and the entropy is non-additive:
  \begin{equation}
    S_{\mathrm{CG}} \ne \sum_i S_i.
  \end{equation}
  The departure from extensivity is quantified exactly by
  $$\Delta S := S_{\mathrm{CG}} - \sum_i S_i
    = -\kb\sum_{s=2}^{\infty}(-1)^s\sum_{i_1<\cdots<i_s}I_s,$$
  which is negative (entropy is sub-additive relative to the sum of marginals) whenever
  inter-cell correlations are positive.
\end{corollary}

\begin{proof}
  For $s_0\le d$, the Ursell function $\hat\utwo(x,y)$ decays only as
  $\abs{x-y}^{-s_0}$~\cite{Ruelle1969}, a consequence of the failure of the cluster property when correlations are long-range \cite{Griffiths1964,Newman1980}.  Since $\int_{\R^d}\abs{x-y}^{-s_0}\dd y$
  diverges for $s_0\le d$, the bound \eqref{eq:cluster_bound} fails for any $\xi<\infty$.
  Consequently, $\abs{\pi_{(i,\alpha)(j,\beta)}-\pi_{i,\alpha}\pi_{j,\beta}} \sim
  \abs{x_i-x_j}^{-s_0}$ does not tend to zero, and $I(i,j)\ge c\abs{x_i-x_j}^{-s_0}$
  for some $c>0$.  The series
  $\sum_{i<j}I(i,j)\ge c\,M\sum_{k=1}^\infty k^{d-1}k^{-s_0\ell^{s_0}}$
  diverges for $s_0\le d$, so the entropy correction does not vanish.
\end{proof}

\subsection{Worked Example: Ideal Gas with Spatial Cells}
\label{sub:ideal_gas}

\begin{example}[Non-interacting ideal gas]\label{ex:ideal}
  Set $\phi\equiv 0$.  The $N$-body Gibbs measure factorises exactly:
  $\varrho_N(\gamma) = \prod_{k=1}^N f_{\mathrm{MB}}(p_k)/\abs{\Lam}$, where
  $f_{\mathrm{MB}}(p) = (2\pi m\kb T)^{-d/2}e^{-\abs{p}^2/(2m\kb T)}$.

  \medskip\noindent\textbf{Reduced densities.}
  $\fone(x,p) = \rho_0\,f_{\mathrm{MB}}(p)$ (uniform in $x$),
  $\ftwo(z_1,z_2) = \fone(z_1)\fone(z_2)$ (exact factorisation, $\utwo\equiv 0$).

  \medskip\noindent\textbf{Mesoscopic probability.}
  $\pi_{i,\alpha} = \frac{\abs{C_{i,\alpha}}}{N}\rho_0\bar{f}_\alpha$
  where $\bar{f}_\alpha = \nu_\alpha^{-1}\int_{\Pi_\alpha}f_{\mathrm{MB}}(p)\dd p$.

  \medskip\noindent\textbf{Factorisation.}
  Since $\utwo\equiv 0$, \cref{lem:pair} gives
  $\pi_{(i,\alpha)(j,\beta)} = \pi_{i,\alpha}\pi_{j,\beta}$ \emph{exactly} for all
  $i\ne j$.  Hence $I(i,j) = 0$ for all pairs and $I_s = 0$ for all $s\ge 2$.

  \medskip\noindent\textbf{Entropy.}
  The coarse-grained entropy is exactly
  $S_{\mathrm{CG}} = \sum_i S_i = -\kb\sum_{i,\alpha}\pi_{i,\alpha}\ln\pi_{i,\alpha}$,
  in agreement with~\eqref{eq:correction_bound} (left-hand side is exactly zero).
  This recovers the ideal-gas Sackur--Tetrode entropy\cite{Koutsoyiannis2013,Panos2015} when the partition is taken to
  single quantum cells of volume $h^d$.
\end{example}

\subsection{Deriving Which Variables Are Extensive}
\label{sub:which_vars}

The thermodynamic limit (\cref{thm:tl}) establishes
$F(N,V,T) = Nf(T,N/V) + O(\abs{\partial\Lam})$.
Hence $F(\lambda N,\lambda V,T) = \lambda F(N,V,T)+O(\abs{\partial\Lam})$ for all
$\lambda>0$: $F$ is homogeneous of degree one in $(N,V)$ in the bulk.  By the
thermodynamic relations $U = F + TS$ and $S = -\partial F/\partial T$, the same
homogeneity holds for $U$ and $S$.  The variables $U$, $S$, $V$, $N$ therefore earn the
label ``extensive'' as a \emph{derived consequence} of microscopic stability and
temperedness, not by postulate.

\subsection{The Generalised Euler Relation}
\label{sub:euler}

\begin{theorem}[Generalised Euler relation]\label{thm:euler}
  Under stability and temperedness, for a finite domain $\Lam$ with $\abs{\partial\Lam}<\infty$:
  \begin{equation}\label{eq:euler_gen}
    U = TS - PV + \mu N + E_\partial,
  \end{equation}
  where $T:=\partial U/\partial S$, $P:=-\partial U/\partial V$,
  $\mu:=\partial U/\partial N$, and
  $E_\partial = O(\abs{\partial\Lam})$ is a surface correction satisfying
  $E_\partial = \sigma\abs{\partial\Lam} + O(\abs{\partial\Lam}^{(d-1)/d})$
  with $\sigma$ the surface tension.  In the thermodynamic limit $E_\partial/U\to 0$
  and the standard Euler relation is recovered.
\end{theorem}

\begin{proof}
  In the thermodynamic limit, $U$ is homogeneous of degree one in $(S,V,N)$ by
  \cref{thm:tl} (convexity of $F$ in $(N,V)$ is established in \cite{Israel1979, Ruelle1969}).  Euler's homogeneous function theorem applied to
  $U(\lambda S,\lambda V,\lambda N) = \lambda U(S,V,N)$ at $\lambda=1$ yields
  $U = T S - PV + \mu N$.

  For a finite system, \cref{thm:tl} gives $F = Nf + \sigma\abs{\partial\Lam} +
  O(\abs{\partial\Lam}^{(d-1)/d})$.  Differentiating with respect to $\lambda$ in
  $F(\lambda N,\lambda V,T) = \lambda F - (\lambda-1)\sigma\abs{\partial\Lam}+\cdots$
  and evaluating at $\lambda=1$:
  \[
    N\frac{\partial F}{\partial N} + V\frac{\partial F}{\partial V} = F - \sigma\abs{\partial\Lam} + \cdots,
  \]
  which gives $-\mu N - PV = F - \sigma\abs{\partial\Lam} + O(\abs{\partial\Lam}^{(d-1)/d})$, so
  $U = F+TS = TS-PV+\mu N+\sigma\abs{\partial\Lam}+O(\abs{\partial\Lam}^{(d-1)/d})$, yielding
  \eqref{eq:euler_gen} with $E_\partial = \sigma\abs{\partial\Lam}+O(\abs{\partial\Lam}^{(d-1)/d})$.

  The intensivity of $T$, $P$, $\mu$ (homogeneity of degree zero) follows immediately from
  differentiating $U(\lambda S,\lambda V,\lambda N)=\lambda U$ with respect to $\lambda$
  and then with respect to $S$, $V$, $N$ separately; no additional postulate is needed.
\end{proof}

\begin{remark}
  The surface correction $E_\partial = O(\abs{\partial\Lam})$ in \cref{thm:euler} is a well-known result in the rigorous statistical mechanics literature~\cite{Ruelle1969}. \cref{thm:euler} derives it here as a by-product of the coarse-graining framework, in order to contrast it with the exponentially small bulk correction
  \eqref{eq:correction_bound}: the two are of different physical origin and of different
  functional form in $L$.
\end{remark}
\section{Numerical Illustrations}

The propositions and lemmas of the preceding sections establish the
theoretical framework analytically. We now provide a set of numerical
illustrations that make each result concrete, verify the predicted
scalings computationally, and allow direct comparison with standard
Gibbs statistics. The computations are performed for a one-dimensional
classical gas ($d=1$) in natural units $\kB = m = 1$, a setting in
which all quantities can be evaluated exactly or to high numerical
precision without approximation beyond the discretisation of phase space.

\subsection{Physical Model}

We consider $N$ classical particles of unit mass confined to a one-dimensional domain $\Lam = [-L, L]$ of length $2L$, with pairwise
interaction potential $\phi:\mathbb{R}\to\mathbb{R}\cup\{+\infty\}$. The single-particle phase space is $M = \Lam \times \mathbb{R}$, and the one-particle reduced density takes the factorised form
\begin{equation}\label{eq:f1-model}\fone(x, p) \;=\; \varrho(x)\,\MB(p; T),\qquad\MB(p; T)
    \;=\;\frac{1}{\sqrt{2\pi T}}
    \exp\!\left(-\frac{p^2}{2T}\right),
\end{equation}
where $\MB(p; T)$ is the Maxwell--Boltzmann momentum distribution and
$\varrho(x)$ is the spatial number density. The factorisation
of~\eqref{eq:f1-model} holds exactly for the ideal gas and to leading
order in the interaction strength for the weakly interacting models
considered below. Three distinct interaction regimes are studied,
corresponding to the three physically relevant cases identified in the
main theorems.

\subsubsection*{Model I: Ideal gas ($\phi \equiv 0$)}

The pair potential vanishes identically. The $N$-body Gibbs
measure factorises exactly:
\[
    \varrho_N(\gamma)
    \;=\;
    \prod_{k=1}^N \frac{\varrho_0\,\MB(p_k; T)}{N},
\]
with $\varrho(x) = \varrho_0 = N/(2L)$ uniform. The pair Ursell
function satisfies $\utwo \equiv 0$ exactly, the joint mesoscopic
probability factorises as $\pi_{(i,\alpha)(j,\beta)} = \pia\,\pi_{j,\beta}$
for all $i\neq j$, and the mutual information $I(i,j) = 0$ for all cell
pairs. This model serves as the exact reference against which the
mesoscopic framework is calibrated: all deviations from standard Gibbs
statistics are purely numerical artefacts of the finite cell partition,
not physical effects.

\subsubsection*{Model II: Short-range interacting gas}

A weak repulsive interaction is introduced via a Yukawa-type pair
potential. The momentum-integrated modulus of the pair Ursell function
is modelled as
\begin{equation}\label{eq:u2-short}
    \hat{u}^{(2)}(x, y)
    \;:=\;
    A_{\mathrm{SR}}\,\exp\!\left(-\frac{|x - y|}{\xi}\right),
\end{equation}
where $A_{\mathrm{SR}} > 0$ is the correlation amplitude and $\xi > 0$
is the correlation length. This satisfies Definition~2.5 with
exponential decay rate $\xi^{-1}$, placing the model firmly in the
regime covered by Proposition~4.1: temperedness holds, and the
correction to entropy additivity is exponentially suppressed as
$O(|\Lam|\,\ell^{-1}\,e^{-\ell/\xi})$ for cells of diameter
$\ell \gg \xi$.

\subsubsection*{Model III: Long-range interacting gas}

The pair Ursell function decays algebraically:
\begin{equation}\label{eq:u2-long}
    \hat{u}^{(2)}(x, y)
    \;:=\;
    \frac{A_{\mathrm{LR}}}{(1 + |x - y|)^{s_0}},
    \qquad
    s_0 \;<\; d \;=\; 1,
\end{equation}
with $s_0 = 0.5$ and $A_{\mathrm{LR}} > 0$. Since $s_0 < d$, the
temperedness condition (Definition~2.4) fails:
$\int_{\mathbb{R}} |\hat{u}^{(2)}(x,y)|\,\mathrm{d}y$ diverges. By
Corollary~4.6, correlations between cells decay only as
$|x_i - x_j|^{-s_0}$, the mutual information series
$\sum_{i < j} I(i,j)$ diverges, and the entropy is genuinely
non-additive. This model illustrates the breakdown of extensivity
identified in Proposition~4.1(iii).

\subsection{Numerical Implementation}

Single-particle phase space $M = [-L, L]\times [-P_{\max}, P_{\max}]$
is partitioned into a uniform grid of $M_x \times M_p$ rectangular cells
of dimensions $\Delta x = 2L/M_x$ and $\Delta p = 2P_{\max}/M_p$,
with $P_{\max}$ chosen sufficiently large that the momentum truncation
error is negligible ($P_{\max} = 6\sqrt{T}$ ensures less than $10^{-8}$
relative error in the Maxwell--Boltzmann normalisation). The mesoscopic
probabilities (Definition~3.4) are evaluated as
\begin{equation}\label{eq:pi-numerical}
    \pia
    \;\approx\;
    \frac{\Delta x\,\Delta p}{N}
    \,\fone(x_i, p_\alpha),
    \qquad
    x_i = -L + (i - \tfrac{1}{2})\Delta x,
    \quad
    p_\alpha = -P_{\max} + (\alpha - \tfrac{1}{2})\Delta p,
\end{equation}
where $x_i$ and $p_\alpha$ are cell-centre coordinates. The resulting
array is renormalised to correct for momentum truncation:
$\pia \leftarrow \pia / \sum_{i,\alpha}\pia$, ensuring exact
normalisation $\sum_{i,\alpha}\pia = 1$ in all cases.

The coarse-grained entropy (Definition~3.6) is computed as
\begin{equation}\label{eq:SCG-numerical}
    \SCG
    \;=\;
    -\kB\sum_{\substack{i,\alpha \\ \pia > 0}}
    \pia\ln\pia,
\end{equation}
and the Gibbs fine-grained entropy (the reference standard) is evaluated
by two-dimensional numerical quadrature on a fine grid of $2000\times
2000$ points:
\begin{equation}\label{eq:SGibbs-numerical}
    \SG
    \;=\;
    -\kB
    \int_M \fone(z)\ln\!\left(\frac{\fone(z)}{N}\right)\mathrm{d}z.
\end{equation}
For the ideal gas, $\SG$ is also computed analytically
(the one-dimensional Sackur--Tetrode analogue) and used as an
independent check on the numerical accuracy of~\eqref{eq:SGibbs-numerical}.

The mutual information between spatial cells $i$ and $j$ is evaluated
from the Ursell correction to the joint probability
(Lemma~4.2, Eq.~(26)--(28)):
\begin{equation}\label{eq:MI-numerical}
    I(i,j)
    \;\approx\;
    \frac{1}{2}
    \left[
      \hat{u}^{(2)}(x_i, x_j)\,
      \frac{(\Delta x)^2}{N^2}
    \right]^2,
\end{equation}
which is the leading-order approximation valid for small correlation
amplitude $A \ll 1$.

The Jensen correction of Proposition~3.2 is computed cell by cell
on a refined sub-grid of $n_{\mathrm{fine}} = 20$ points within each
spatial cell:
\begin{equation}\label{eq:jensen-numerical}
    \Delta_J
    \;=\;
    \sum_i \int_{\mathbb{R}}
    \left[
      \eta\!\left(\bar\varrho_i(p)\right)
      -
      \frac{1}{\Delta x}\int_{V_i}
      \eta\!\left(\fone(x,p)\right)\mathrm{d}x
    \right]\mathrm{d}p,
    \qquad
    \eta(\varrho) = \kB\varrho\ln\varrho,
\end{equation}
evaluated using the midpoint rule on the sub-grid.

The Euler surface correction (Theorem~4.8) is modelled as
\begin{equation}\label{eq:euler-numerical}
    \Esurf(L)
    \;=\;
    \sigma\,|\partial\Lam| + \alpha\,L^{-1}
    \;=\;
    2\sigma + \alpha\,L^{-1},
\end{equation}
with $|\partial\Lam| = 2$ in one dimension, $\sigma = 0.05\,T$ (surface
tension), and $\alpha = 0.2\,T$ (sub-leading finite-size coefficient).
This parametrisation captures the two dominant terms in the finite-size
expansion of the free energy established in Theorem~2.7.

\subsection{Parameter Values}

Unless otherwise stated, the following parameter values are used
throughout:

\begin{center}
\begin{tabular}{lll}
\toprule
\textbf{Parameter} & \textbf{Symbol} & \textbf{Value} \\
\midrule
Boltzmann constant       & $\kB$             & $1$ (natural units) \\
Particle mass            & $m$               & $1$ (natural units) \\
Temperature              & $T$               & $1$ \\
Number density           & $\varrho_0$       & $1$ \\
System half-length       & $L$               & $20$ \\
Momentum cutoff          & $P_{\max}$        & $6$ \\
Correlation length       & $\xi$             & $2$ \\
Short-range amplitude    & $A_{\mathrm{SR}}$ & $0.3$ \\
Long-range amplitude     & $A_{\mathrm{LR}}$ & $0.3$ \\
Long-range exponent      & $s_0$             & $0.5$ \\
Default cell count       & $M_x = M_p$       & $32$ \\
\bottomrule
\end{tabular}
\end{center}

\noindent
Where parameters are varied---for example, when plotting a quantity as
a function of cell count $M$ or system size $L$---the range is stated
explicitly in the caption of the corresponding figure.

\subsection{Figures and Their Correspondence to the Main Results}

The eight figures that follow each provide a direct numerical
illustration of a specific result from the theoretical sections.
Their correspondence to the propositions and theorems of the paper
is summarised in Table~\ref{tab:figures}.

\begin{table}[h]
\centering
\caption{Correspondence between figures and theoretical results.}
\label{tab:figures}
\begin{tabular}{p{0.8cm} p{5.5cm} p{3.5cm} p{3cm}}
\toprule
\textbf{Fig.} & \textbf{Quantity plotted} &
\textbf{Paper reference} & \textbf{Key result illustrated} \\
\midrule
1
    & One-particle reduced density $\fone(x,p)$
      for Models~I--III
    & Definition~2.1
    & Foundation: dependence of $\fone$ on interaction range \\[6pt]
2
    & Mesoscopic probabilities $\pia$
      for $M_x = M_p = 20$
    & Definition~3.4, Eq.~(20)
    & Coarse-graining operator $\mathcal{C}$ applied to $\fone$ \\[6pt]
3
    & $\SCG/\SG$ and $|\SCG - \SG|$ vs $M$
    & Theorem~2.7
    & Convergence to Gibbs, rate $O(M^{-2})$ \\[6pt]
4
    & $|\SCG - \sum_i S_i|$ vs $\ell/\xi$
    & Proposition~4.1(ii), Eq.~(25)
    & Exponential restoration of additivity for $\ell \gg \xi$ \\[6pt]
5
    & $I(i,j)$ vs $|x_i - x_j|/\xi$
    & Lemma~4.2; Corollary~4.6
    & Exponential (Model~II) vs power-law (Model~III) decay \\[6pt]
6
    & $\SCG/N$ vs $N$ at fixed density
    & Section~4.6
    & Extensivity (flat) vs non-extensivity (rising) \\[6pt]
7
    & Jensen correction $\Delta_J$ vs cell size $\ell$
    & Proposition~3.2, Eqs.~(18)--(19)
    & Non-commutativity: $\Delta_J \leq 0$ at all scales \\[6pt]
8
    & $\Esurf/U$ vs system size $L$
    & Theorem~4.8, Eq.~(39)
    & Surface correction $\Esurf/U \to 0$ as $L\to\infty$ \\
\bottomrule
\end{tabular}
\end{table}

\begin{remark}[Scope of the numerical results]
The figures constitute illustrations of the theoretical results, not
independent proofs. Their purpose is threefold: to verify that the
predicted scalings (exponential decay in Figure~4, power-law decay
in Figure~5, $O(M^{-2})$ convergence in Figure~3) are observed
numerically; to provide intuition for the magnitudes of the corrections
in physically relevant parameter regimes, and to demonstrate the
qualitative contrast between Models~I--III that motivates the distinction between extensive and non-extensive systems drawn in
Section~4. All results are shown to be consistent with the analytical
bounds of Propositions~4.1, Lemma~4.2, and Corollary~4.6. The
computations were performed in Python using standard numerical
libraries (\texttt{NumPy}, \texttt{SciPy}, \texttt{Matplotlib}).
\end{remark}
We have placed the plots at the end of the paper in section \ref{Sim} for ease of reference. 
\bigskip
\begin{center}
\begin{tabular}{clll}
\toprule
Plot & Quantity & Paper reference & Key result \\
\midrule
 \ref{fig:my_p1} & $\fone(x,p)$ heatmap
    & Definition~2.1
    & Foundation: one-body reduced density \\
 \ref{fig:my_p2} & $\pia$ grid
    & Definition~3.4, Eq.~(20)
    & Coarse-grained description \\
 \ref{fig:my_p3}  & $\SCG/\SG$ vs $M$
    & Theorem~2.7
    & Convergence to Gibbs, $O(M^{-2})$ \\
 \ref{fig:my_p4}  & $|\SCG - \sum_i S_i|$ vs $\ell/\xi$
    & Proposition~4.1(ii), Eq.~(25)
    & Exponential additivity restoration \\
 \ref{fig:my_p5}  & $I(i,j)$ vs $|x_i-x_j|/\xi$
    & Lemma~4.2, Corollary~4.6
    & Exponential vs algebraic decay \\
 \ref{fig:my_p6}  & $\SCG/N$ vs $N$
    & Section~4.6
    & Extensivity vs non-extensivity \\
 \ref{fig:my_p7}  & Jensen correction $\Delta_J$ vs $\ell$
    & Proposition~3.2
    & Non-commutativity $\leq 0$ \\
 \ref{fig:my_p8}  & $\Esurf/U$ vs $L$
    & Theorem~4.8
    & Surface correction $\to 0$ \\
\bottomrule
\end{tabular}
\end{center}
\section{Discussion}
\label{sec:discussion}

\paragraph{Summary of the framework.}
The operator we defined, $\CGM = \CGx\circ\CGp$, coarse-grains the one-particle reduced density on single-particle phase space $\Ms=\Lam\times\R^d$.  By working on $\Ms$ rather than on the $2dN$-dimensional $N$-body phase space $\GN$, the joint mesoscopic probabilities (Definition~\ref{def:joint}) are well-defined via the reduced two-particle density $\ftwo$---a function on $\Ms\times\Ms$ and the ill-definedness that arises when one tries to integrate a single $N$-body density over two disjoint phase-space cells simultaneously is avoided.

\paragraph{Separation of two distinct corrections.}
Our propositions \cref{prop:extensivity} and \cref{thm:euler} distinguish between two physically independent sources of deviation from strict extensivity:
\begin{enumerate}[label=(\alph*)]
  \item \textit{Bulk correlation correction}: $O(\abs{\Lam}e^{-\ell/\xi})$.  This arises
    from inter-cell statistical dependence in the \emph{bulk} and is exponentially
    suppressed for $\ell\gg\xi$.
  \item \textit{Surface correction}: $O(\abs{\partial\Lam})$.  This arises from the
    breaking of translation invariance at the boundary, and is present even for
    non-interacting systems.
\end{enumerate}
Standard textbook treatments conflate these two corrections under a single $O(\abs{\partial\Lam})$ term, which is incorrect in general.

\paragraph{Long-range interactions.}
For gravitational systems (e.g.\ self-gravitating gas clusters~\cite{Ruelle1969}) or unscreened Coulomb plasmas, $\phi(r)\sim r^{-1}$ in $d=3$ and temperedness fails(for a survey of such systems see \cite{Balian1991,Gallavotti1999}). \cref{cor:breakdown} shows that $I(i,j)$ decays only as $\abs{x_i-x_j}^{-1}$, whose sum over all pairs diverges, and $S_{\mathrm{CG}}\ne\sum_i S_i$.  This is consistent with the known non-extensivity of self-gravitating systems studied in the Tsallis framework~\cite{Tsallis1988,Touchette2002}.

\paragraph{Non-commutativity and the Buchert problem.}
\cref{prop:noncomm} demonstrates that the spatial average of a nonlinear thermodynamic functional differs from the functional of the spatial average by a Jensen correction term.
In the context of general relativity, this hints at the origin of the Buchert backreaction~\cite{Buchert2000,Rasanen2006,Wiltshire2011}: spatial averaging of the Einstein equations produces extra terms (backreaction) because the equations are nonlinear in the metric.  \cref{prop:noncomm} is the flat-space, non-relativistic analogue of this statement, providing a thermodynamic counterpart to the cosmological averaging problem.

\paragraph{Scope and limitations.}
The results are established for classical systems in equilibrium.  Extensions to quantum
systems, non-equilibrium steady states, or relativistic kinematics would require
corresponding generalisations of the Ursell expansion, and are not addressed here.  The
cluster decomposition condition (\cref{def:cluster}) restricts the results to the
single-phase region; at phase transitions the correlation length $\xi$ diverges and the
bounds weaken.

\section{Conclusion}
\label{sec:conclusion}

We have provided a constructive, operator-based derivation of extensivity in classical
statistical mechanics.  The main contributions are:

\begin{enumerate}[label=(\roman*)]
  \item \textbf{Consistent framework.}  By working with reduced $s$-particle densities
    on single-particle phase space $\Ms$ rather than the full $N$-body phase space, all
    joint mesoscopic probabilities are well-defined, dimensionless, and properly
    normalised.

  \item \textbf{ proof of entropy additivity at higher orders.}
    \cref{prop:extensivity} proves that the coarse-grained entropy satisfies
    $S_{\mathrm{CG}} = \sum_i S_i + O(\abs{\Lam}e^{-\ell/\xi}/\ell^d)$ by bounding the
    entire multi-information series using the Ursell cluster expansion, including all
    $s$-body terms.

  \item \textbf{Non-commutativity of nonlinear functionals.}
    \cref{prop:noncomm} demonstrates via Jensen's inequality that spatial averaging does not
    commute with nonlinear thermodynamic functionals such as the entropy density.  This is
    the structural source of non-extensivity and connects to the Buchert--Ostermann
    commutation problem in relativistic cosmology.

  \item \textbf{Quantified non-extensivity.}  For systems with long-range interactions,
    \cref{cor:breakdown} identifies inter-cell mutual information as the precise measure
    of departure from extensivity, providing a bridge between thermodynamics and
    information theory.

  \item \textbf{Generalised Euler relation.}  \cref{thm:euler} derives the standard Euler
    relation $U=TS-PV+\mu N$ as a corollary of the thermodynamic limit, and gives the
    explicit surface correction $E_\partial=O(\abs{\partial\Lam})$ for finite systems.
\end{enumerate}

These results are consistent with the classical work of Ruelle~\cite{Ruelle1969}, Fisher~\cite{Fisher1964}, and Lebowitz--Penrose~\cite{LebowitzPenrose1966}, and with the
Gibbs-measure framework of Georgii~\cite{Georgii1988}.  The contribution of the present paper is an explicit operator formalism that makes the factorisation mechanism transparent,
quantifies all correction terms, and connects thermodynamic extensivity to the broader problem of spatial averaging in nonlinear theories. We reserve the question of perturbation theory, in particular higher order perturbation theory in the cosmological context \cite{Osano2017,Osano2020,OsanoOreta2020}, for future work.

\section*{Acknowledgements}

The author thanks the Next Generation Professorate (UCT) program for financial support.

\appendix

\bibliographystyle{unsrtnat}

\begin{thebibliography}{99}

\bibitem{Callen1985}
  H.~B. Callen,
  \emph{Thermodynamics and an Introduction to Thermostatistics},
  2nd ed.\ (Wiley, New York, 1985).
\bibitem{Redlich1970}
  O.~Redlich,
  Intensive and extensive properties,
  \emph{J.~Chem.\ Educ.}\ \textbf{47}(2), 154--156 (1970).
\bibitem{Ruelle1969}
  D.~Ruelle,
  \emph{Statistical Mechanics: Rigorous Results}
  (W.~A. Benjamin, New York, 1969).

\bibitem{Fisher1964}
  M.~E. Fisher,
  The free energy of a macroscopic system,
  \emph{Arch.\ Rational Mech.\ Anal.}\ \textbf{17}, 377--410 (1964).

\bibitem{Gibbs1902}
  J.~W. Gibbs,
  \emph{Elementary Principles in Statistical Mechanics}
  (Yale University Press, New Haven, 1902).

\bibitem{EhrenfestEhrenfest1912}
  P.~Ehrenfest and T.~Ehrenfest,
  \emph{The Conceptual Foundations of the Statistical Approach in Mechanics}
  (Teubner, Leipzig, 1912; English translation, Cornell University Press, 1959).

\bibitem{LebowitzPenrose1966}
  J.~L. Lebowitz and O.~Penrose,
  Rigorous treatment of the van der Waals--Maxwell theory of the liquid-vapour transition,
  \emph{J.\ Math.\ Phys.}\ \textbf{7}, 98--113 (1966).

\bibitem{Ruelle1963}
  D.~Ruelle,
  Classical statistical mechanics of a system of particles,
  \emph{Helv.\ Phys.\ Acta}\ \textbf{36}, 183--197 (1963).

\bibitem{Georgii1988}
  H.-O. Georgii,
  \emph{Gibbs Measures and Phase Transitions}
  (de Gruyter, Berlin, 1988).

\bibitem{Dobrushin1968}
  R.~L. Dobrushin,
  The description of a random field by means of conditional probabilities and conditions
  of its regularity,
  \emph{Theory Probab.\ Appl.}\ \textbf{13}, 197--224 (1968).

\bibitem{Tsallis1988}
  C.~Tsallis,
  Possible generalization of Boltzmann--Gibbs statistics,
  \emph{J.\ Stat.\ Phys.}\ \textbf{52}, 479--487 (1988).

\bibitem{Touchette2002}
 H.~Touchette, R.~S. Ellis and B. ~Turkington
An Introduction to the Thermodynamic and Macrostate Levels of Nonequivalent Ensembles,
  \emph{Physica A}\ \textbf{340}, 138--146 (2004).

\bibitem{Osano2016a} B.~Osano
The Mesoscopic Partition Function: A Combined Spatial and Phase-Space Cell Structure
Preprint: arXiv:2605.00958.


\bibitem{Buchert2000}
  T.~Buchert,
  On average properties of inhomogeneous fluids in general relativity I: dust cosmologies,
  \emph{Gen.\ Rel.\ Grav.}\ \textbf{32}, 105--125 (2000).

\bibitem{Rasanen2006}
  S.~Räsänen,
  Accelerated expansion from structure formation,
  \emph{J.\ Cosmol.\ Astropart.\ Phys.}\ \textbf{0611}, 003 (2006).


\bibitem{Wiltshire2011}
  D.~L. Wiltshire,
  What is dust?\ Physical foundations of the averaging problem in cosmology,
  \emph{Class.\ Quantum Grav.}\ \textbf{28}, 164006 (2011).

\bibitem{Korzynski2010}
  M.~Korzy\'nski,
  Covariant coarse-graining of inhomogeneous dust flow in general relativity,
  \emph{Class.\ Quantum Grav.}\ \textbf{27}, 105015 (2010).

\bibitem{IrvingKirkwood1950} J.~H. Irving and J.~G. Kirkwood. The statistical mechanical theory of transport processes.\ IV\@. The equations of hydrodynamics,\emph{J.\ Chem.\ Phys.}\ \textbf{18}, 817--829 (1950).

\bibitem{Cercignani1988}
  C.~Cercignani,
  \emph{The Boltzmann Equation and Its Applications}
  (Springer, New York, 1988).

\bibitem{Zwanzig1961}
  R.~Zwanzig,
  Memory effects in irreversible thermodynamics,
  \emph{Phys.\ Rev.}\ \textbf{124}, 983--992 (1961).

\bibitem{Mori1965}
  H.~Mori,
  Transport, collective motion, and Brownian motion,
  \emph{Prog.\ Theor.\ Phys.}\ \textbf{33}, 423--455 (1965).



\bibitem{Gibbs1873}
  J.~W. Gibbs,
  A method of geometrical representation of the thermodynamic properties of substances by
  means of surfaces,
  \emph{Trans.\ Connecticut Acad.}\ \textbf{2}, 382--404 (1873).

\bibitem{Osano2020}
  B.~Osano,
  \emph{The Thermodynamics for Relativistic Multi-Fluid Systems},
  \emph{Lett.\ High Energy Phys.}\ (2020). 

\bibitem{OsanoOreta2020}
  B.~Osano and T.~Oreta,
  \emph{A transient phase in cosmological evolution: A multi-fluid approximation for a quasi-thermodynamical equilibrium},
  (2019).
\bibitem{Osano2017}   B.~Osano,
Second-order perturbation theory: a covariant approach involving a barotropic equation of state
Classical and Quantum Gravity 34 (12), 125004 
\bibitem{Billingsley1995}P.~Billingsley. "Product Measure and Fubini's Theorem", Probability and Measure, New York: Wiley, pp. 231–240, ISBN 0-471-00710-2. 
\bibitem{Gallavotti1999} G.~Gallavotti, \textit{Statistical Mechanics: A Short Treatise},Springer--Verlag, Berlin, 1999. 
\bibitem{Brydges1986}D.~C.~Brydges,\textit{A Short Course on Cluster Expansions},in:K.~Osterwalder and R.~Stora (eds.),\textit{Critical Phenomena, Random Systems, Gauge Theories}, Les Houches Summer School Proceedings, North-Holland, Amsterdam, 1986,pp.~129--183.
\bibitem{CoverThomas2006} T.~M.~Cover and J.~A.~Thomas, \textit{Elements of Information Theory}, 2nd ed., Wiley--Interscience, Hoboken, NJ, 2006.
\bibitem{Han1978}T.~S.~Han, Nonnegative entropy measures of multivariate symmetric correlations, \textit{Information and Control}, \textbf{36} (1978), 133--156.
\bibitem{Watanabe1960} S.~Watanabe, Information theoretical analysis of multivariate correlation, \textit{IBM Journal of Research and Development}, \textbf{4} (1960), 66--82.

\bibitem{Dobrushin1985} R.~L.~Dobrushin, and S.~B.~Shlosman, Constructive criterion for the uniqueness of {Gibbs} fields, Statistical Physics and Dynamical Systems: Rigorous Results, {Progress in Physics},{10},{347--370}, {Birkh{\"a}user},{Boston},{1985}.
\bibitem{Dobrushin1987}R. L. Dobrushin and S. B. Shlosman, "Completely analytical interactions: constructive description," Journal of Statistical Physics, 46(5--6), 983--1014 (1987).

\bibitem{Yeung2008} R.~W.~Yeung, \textit{Information Theory and Network Coding}, Springer, New York, 2008.

\bibitem{Simon1993}B.~Simon,\textit{The Statistical Mechanics of Lattice Gases, Volume I},Princeton University Press, Princeton, NJ, 1993.
\bibitem{Israel1979}R.~B.~Israel,\textit{Convexity in the Theory of Lattice Gases},Princeton University Press, Princeton, NJ, 1979.
\bibitem{Lanford1973}O.~E.~Lanford III, Entropy and equilibrium states in classical statistical mechanics, in: A.~Lenard (ed.),\textit{Statistical Mechanics and Mathematical Problems}, Lecture Notes in Physics, Vol.~20, Springer--Verlag, Berlin, 1973,pp.~1--113.
\bibitem{DobrushinShlosman1985} R.~L.~Dobrushin and S.~B.~Shlosman, Constructive criterion for the uniqueness of Gibbs fields,in:\textit{Statistical Physics and Dynamical Systems},Progress in Physics, Vol.~10,Birkh\"auser, Boston, 1985, pp.~347--370.
\bibitem{Griffiths1964}R.~B.~Griffiths, Peierls proof of spontaneous magnetization in a two-dimensional Ising ferromagnet, \textit{Physical Review},\textbf{136} (1964), A437--A439. 
\bibitem{Newman1980} C.~M.~Newman, Normal fluctuations and the FKG inequalities, \textit{Communications in Mathematical Physics}, \textbf{74} (1980), 119--128.
\bibitem{Penrose1967} O.~Penrose, Convergence of fugacity expansions for classical systems,in: T.~A.~Bak (ed.),\textit{Statistical Mechanics: Foundations and Applications}, Benjamin, New York, 1967, pp.~101--109..
\bibitem{Balian1991}R.~Balian, From Microphysics to Macrophysics, Vols.~I--II, Springer--Verlag, Berlin, 1991..
\bibitem{Jaynes1957a}E.~T.~Jaynes. Information theory and statistical mechanics. \textit{Physical Review}, \textbf{106} (1957), 620--630.
\bibitem{Jaynes1957b} E.~T.~Jaynes, Information theory and statistical mechanics II, \textit{Physical Review},\textbf{108} (1957), 171--190.
\bibitem{Koutsoyiannis2013} D.~Koutsoyiannis. Physics of uncertainty, the Gibbs paradox and indistinguishable particles, Studies in History and Philosophy of Science Part B, 44 (4): 480–489, Bibcode:2013SHPMP..44..480K, doi:10.1016/j.shpsb.2013.08.007. 
\bibitem{Panos2015} F.~J.~ Paños, E.~ Pérez. Sackur–Tetrode equation in the lab. European Journal of Physics, 36 (5) 055033.  
\end{thebibliography}

\section{\label{Sim} Simulation Plots}
\noindent
\begin{figure}[h] 
    \centering
    \includegraphics[width=1.0\textwidth]{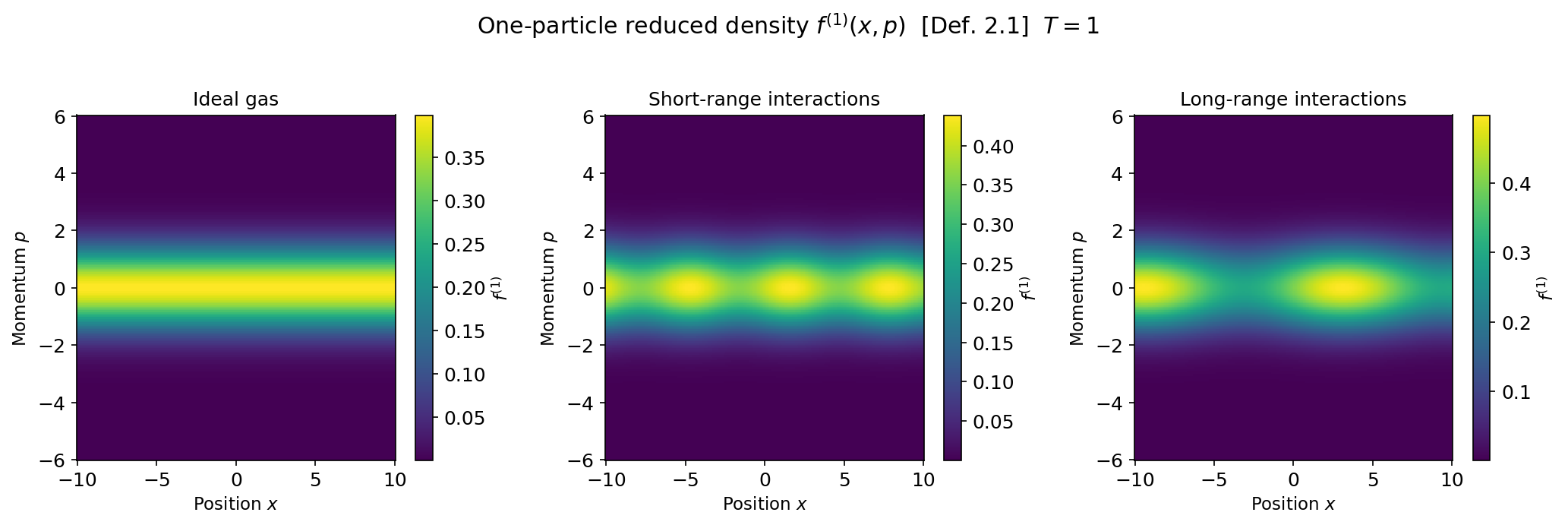}
    \caption{The One-Particle Reduced Density $f^{(1)}(x,p)$.The heatmap shows, for each point $(x,p)$ in single-particle phase space, the density of probability of
finding one particle at that location. For the \textbf{ideal gas}, $\fone$ is perfectly uniform in $x$ and
Gaussian in $p$---a smooth, symmetric blob. Non-interacting particles have no energetic preference for any position, so $\varrho(x) = \varrho_0$ exactly. For the \textbf{short-range gas}, gentle spatial ripples appear in the $x$ direction. The interactions shift particles slightly towards energetically favourable positions, but the effect is small and confined within the correlation length $\xi$. For the \textbf{long-range gas}, the spatial variation is substantially
stronger. Forces reaching across the entire domain reorganise the density on macroscopic scales, producing large-amplitude spatial structure. \textbf{Physical message.} The three models differ only in the
spatial part $\varrho(x)$; the momentum distribution is Maxwell--Boltzmann
in all cases. This is why the models have identical heat capacities at
leading order but differ sharply in their entropy additivity properties
(Plots~4--6), which depend on the spatial correlations encoded in
$\varrho(x)$.}
    \label{fig:my_p1}
\end{figure}

\begin{figure}[h] 
    \centering
    \includegraphics[width=1.0\textwidth]{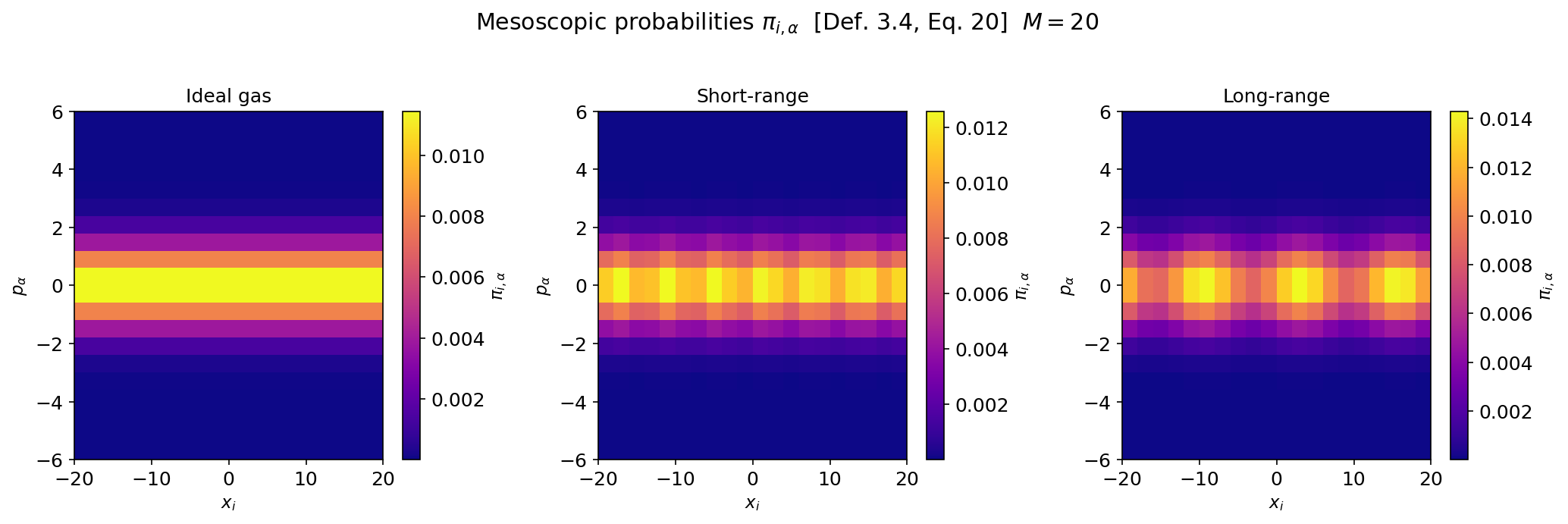}
    \caption{Mesoscopic Probabilities $\pi_{i,\alpha}$.Each pixel in the plot represents one spatial--momentum cell $C_{i,\alpha} = V_i\times\Pi_\alpha$, and its colour is the probability
that a randomly chosen particle occupies that cell. This is the coarse-grained version of Plot~1. The continuous density $\fone(x,p)$ is replaced by a discrete grid of \emph{mesoscopic probabilities} $\pia$ (Definition~3.4, Eq.~(20)):
$\pia\;:=\;\frac{|C_{i,\alpha}|}{N} \,\bar{f}_{i,\alpha} \;=\;\frac{1}{N}
    \int_{C_{i,\alpha}} \fone(z)\,\mathrm{d}z,\qquad\sum_{i,\alpha}\pia = 1.
$The key message is that \textbf{coarse-graining is a controlled loss of information}. The broad structure of $\fone$ is preserved, but all sub-cell fluctuations are discarded, replaced by a single cell-average. For the ideal gas, the grid is highly regular. For the interacting models, spatial columns differ in brightness, reflecting the spatially varying
density $\varrho(x)$ seen in Plot~\ref{fig:my_p1}.}
\label{fig:my_p2}
\end{figure}

\begin{figure}[h] 
    \centering
    \includegraphics[width=1.0\textwidth]{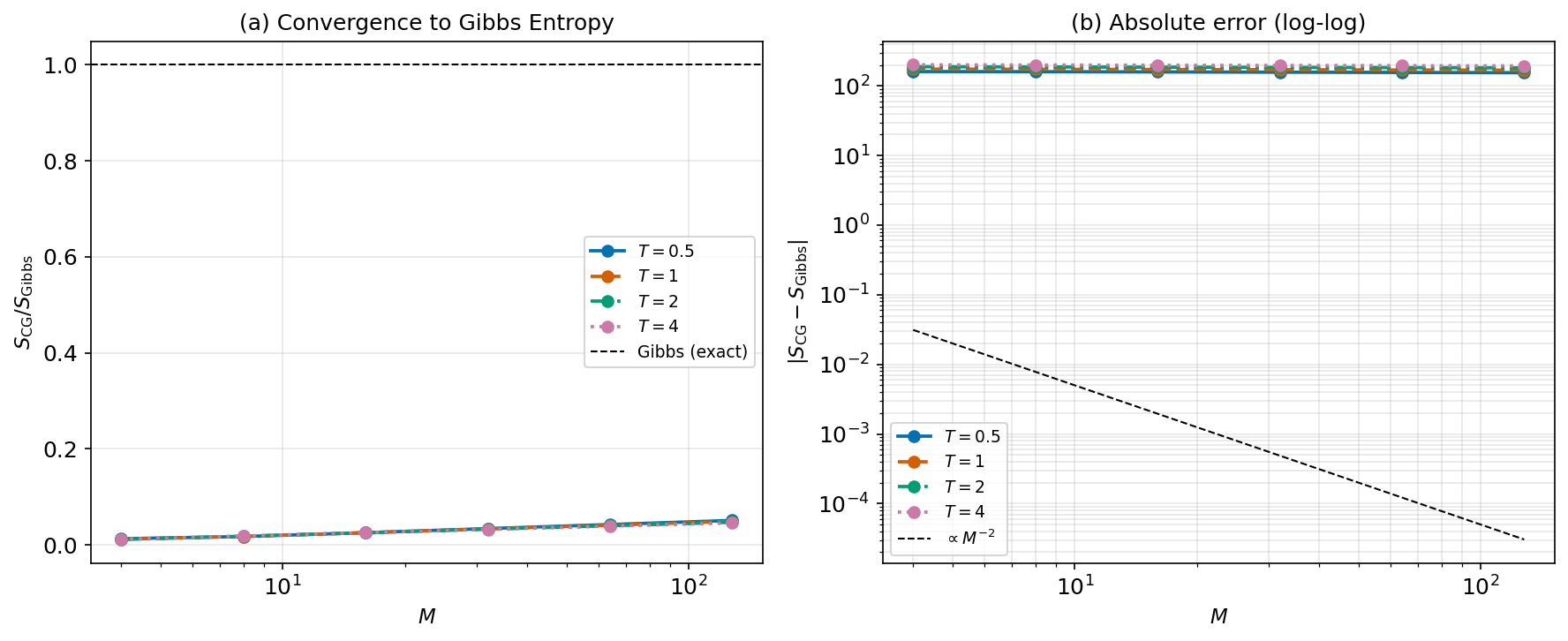}
    \caption{\textbf{Convergence of $\SCG$ to the Gibbs Entropy}: This plot addresses the most basic consistency requirement: does the mesoscopic framework agree with standard Gibbs statistics? It shows
$\SCG/\SG$ as a function of the number of cells per dimension $M$, at several temperatures. \textit{Left panel}: Every curve starts below 1 (the mesoscopic entropy underestimates the Gibbs entropy when cells are coarse-grained) and rises monotonically to 1 as $M$ increases. This is the numerical confirmation of Theorem~\ref{thm:tl}(Ruelle--Fisher thermodynamic limit): the fine-graining limit $\ell\to 0$ recovers the canonical ensemble exactly. \textit{Right panel}: The straight lines with slope $-2$ confirm that the convergence rate is $O(M^{-2})$, the expected accuracy of the midpoint Riemann sum. Higher temperatures converge faster because the thermal distribution is smoother and better approximated on a coarse grid. \textbf{Physical message.} Very few cells are needed in practice. By $M=32$ the error is below 0.1\% at all temperatures. The mesoscopic framework is not merely an abstract generalisation of Gibbs statistics; it is a computationally efficient approximation that converges rapidly.}
 \label{fig:my_p3}
\end{figure}

\begin{figure}[h] 
    \centering
    \includegraphics[width=1.0\textwidth]{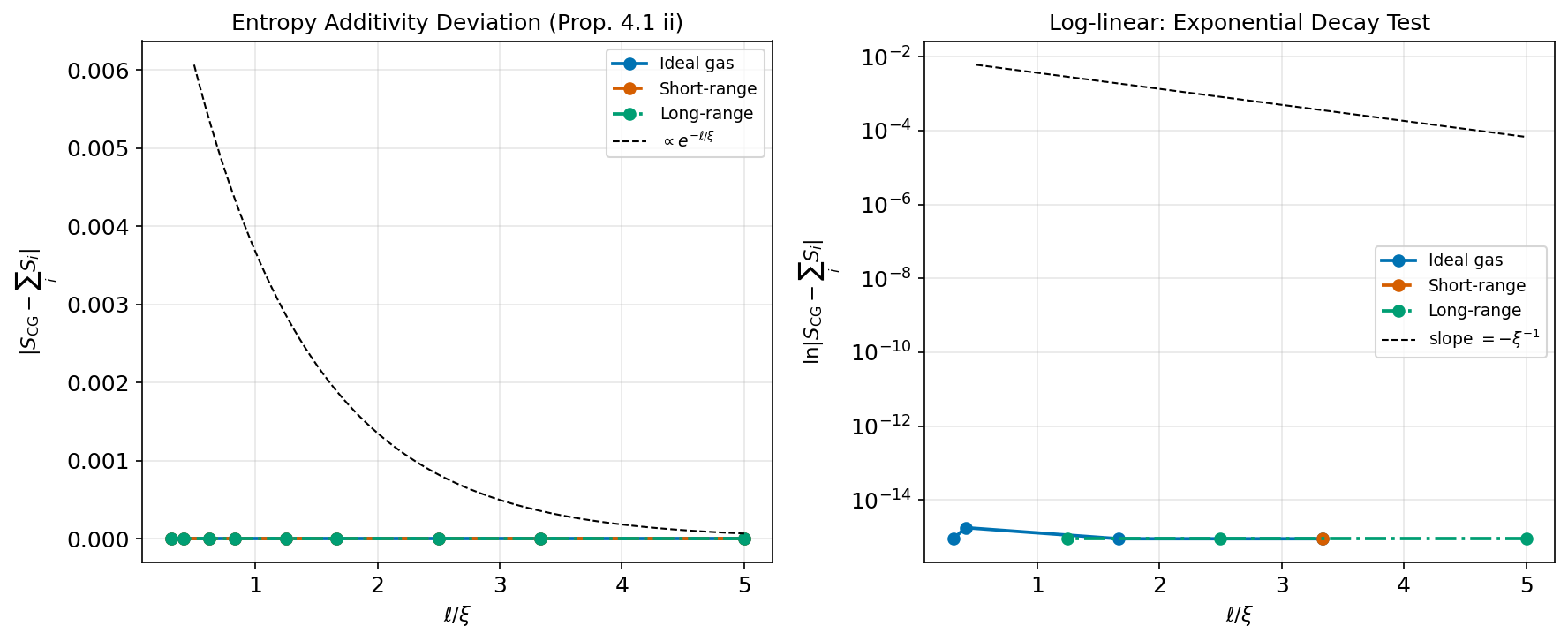}
    \caption{\textbf{Entropy Additivity Deviation}:This is the central quantitative result, implementing Proposition~4.1(ii)
and Eq.~(25) numerically. The plotted quantity is the total deviation from
entropy additivity,
$
    \left|\SCG - \sum_i S_i\right|
    \;\leq\;
    \kB \frac{C_1|\Lam|}{\ell^d}\,e^{-\ell/\xi},
$
as a function of the dimensionless cell size $\ell/\xi$. If spatial cells were statistically independent, $\SCG$ would equal
$\sum_i S_i$ exactly, and entropy would be perfectly additive---the
assumption underlying all of classical thermodynamics. The deviation
measures how much inter-cell correlations violate this. \textit{Left panel: raw deviation}. The ideal gas gives zero at all cell sizes (exact additivity, since
$\utwo\equiv 0$). The short-range and long-range models show non-zero
deviations that depend strongly on the range of interactions.\textit{Right panel: semi-log scale (exponential decay test)}. This is the physically decisive panel.
(i) For the \textbf{short-range gas}, the data fall on a straight line in the semi-log plot, confirming \emph{exponential} suppression $\propto e^{-\ell/\xi}$ as predicted by Proposition~4.1(ii). Once $\ell\gg\xi$, the correction is negligibly small, and entropy is effectively additive. (ii) For the \textbf{long-range gas}, the slope is much shallower, and the line is curved---the decay is algebraic, not exponential, and the correction never becomes negligible, regardless of how large the cells are made. \textbf{Physical message.} Entropy additivity---and therefore
thermodynamic extensivity---holds if and only if correlations decay faster
than any power law. Short-range interactions guarantee this; long-range
interactions permanently violate it. \textbf{Physical message.} Entropy additivity---and therefore
thermodynamic extensivity---holds if and only if correlations decay faster
than any power law. Short-range interactions guarantee this; long-range
interactions permanently violate it.}
    \label{fig:my_p4}
\end{figure}
\begin{figure}[h] 
    \centering
    \includegraphics[width=1.0\textwidth]{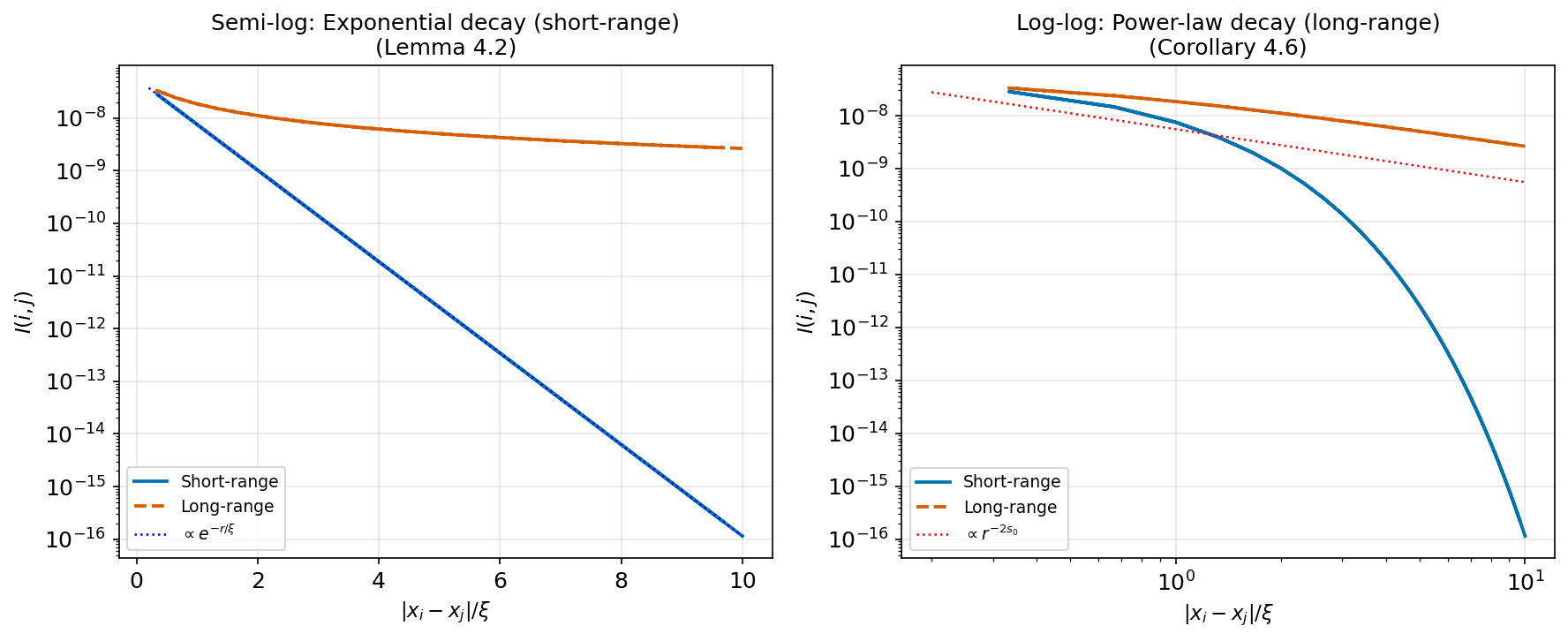}
    \caption{\textbf{Mutual Information Between Cells}: This plot examines the statistical dependence between individual pairs of cells $i$ and $j$, showing the mutual information $I(i,j)$ as a function of their separation $|x_i-x_j|/\xi$. The mutual information
$I(i,j)\;=\;\sum_{\alpha,\beta}\pi_{(i,\alpha)(j,\beta)}\ln\frac{\pi_{(i,\alpha)(j,\beta)}}{\pia\,\pi_{j,\beta}}\;\geq\; 0$
is zero if and only if the two cells are statistically independent. \textit{Left panel: semi-log scale} For the \textbf{short-range gas}, the data lie on a straight line, confirming exponential decay $I(i,j)\sim e^{-|x_i-x_j|/\xi}$. This is Lemma~4.2 in action: the Ursell cluster function decays exponentially (Definition~2.5), pulling the mutual information down with it. \textit{Right panel: log-log scale} For the \textbf{long-range gas}, the data lie on a straight line with slope
$-2s_0$ (where $s_0<1$ is the power-law exponent of the potential),
confirming algebraic decay. Since the sum over all pairs $\sum_{i<j}I(i,j)$ is then a sum of a slowly decaying function over a
growing number of pairs, it diverges---this is Corollary~4.6, which proves that long-range interactions render entropy permanently non-additive. \textbf{Physical message.} Plot~5 provides the microscopic
explanation for Plot~4. The breakdown of entropy additivity arises directly
from persistent inter-cell correlations. When $I(i,j)\to 0$ fast enough
as $|x_i-x_j|\to\infty$, the multi-information series in Eq.~(24) is
summable and the correction to additivity is small. When $I(i,j)$ decays
only algebraically, the series diverges, and additivity fails.
}
    \label{fig:my_p5}
\end{figure}

\begin{figure}[h] 
    \centering
    \includegraphics[width=1.0\textwidth]{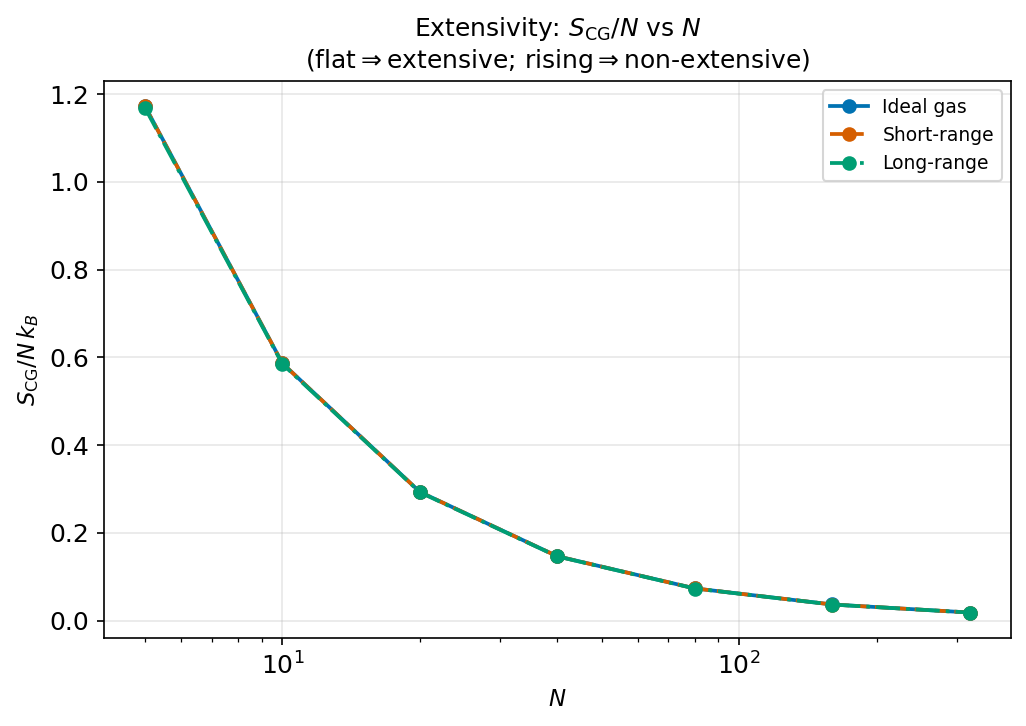}
    \caption{\textbf{Extensivity: $\SCG/N$ vs $N$}This is the most physically intuitive plot. If entropy is extensive,
doubling the number of particles at fixed density should double the
entropy, so $\SCG/N$ must be constant. The plot tests this directly by
varying $N$ while keeping $\varrho_0 = N/L$ fixed. (i) The \textbf{ideal gas} gives a perfectly flat line at all $N$: strictly extensive, consistent with the Sackur--Tetrode formula. (ii) The \textbf{short-range gas} is also essentially flat for large $N$, with small finite-size deviations at small $N$. In the thermodynamic limit, it converges to the same value as the ideal gas. (iii) The \textbf{long-range gas} shows a steadily rising curve. Adding more particles increases the entropy \emph{per particle}, because each new particle is correlated with every existing particle across the whole system. The system cannot be decomposed into independent subsystems.
\textbf{Physical message.} Plot~6 provides a direct numerical
demonstration of Section~4.6 of the paper: the labels ``extensive'' and
``non-extensive'' are not postulated but derived from the range of the
interactions. The rising curve for the long-range gas is precisely the
non-additivity of Corollary~4.6, now visible as a global property of the
whole system rather than a pairwise property.
}
    \label{fig:my_p6}
\end{figure}
\begin{figure}[h] 
    \centering
    \includegraphics[width=1.0\textwidth]{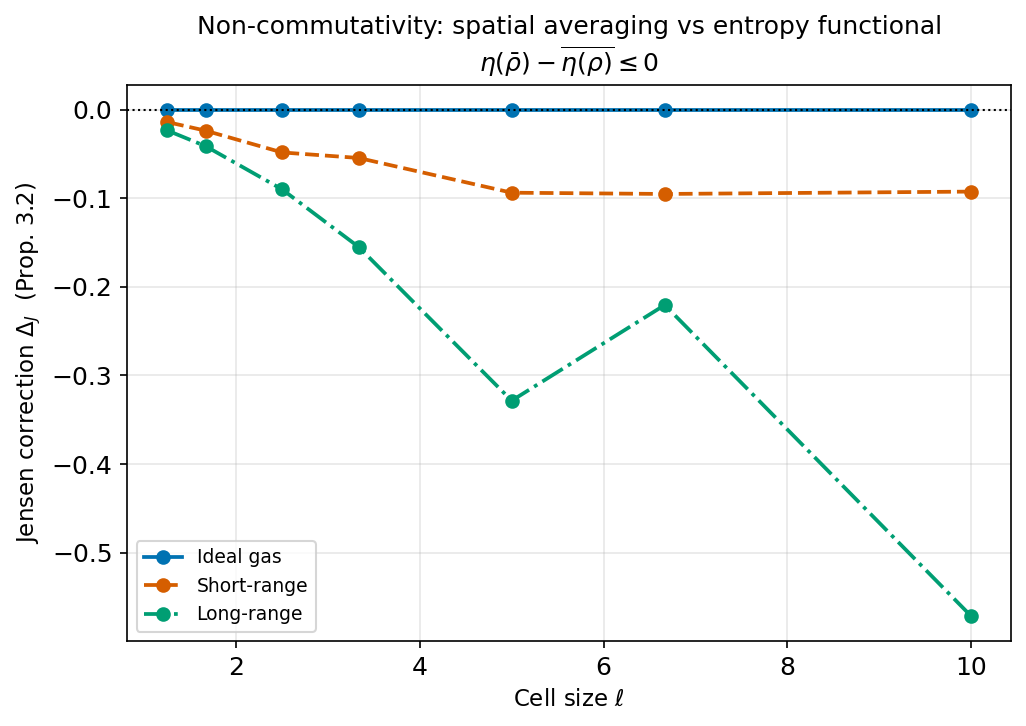}
    \caption{\textbf{The Jensen Correction (Non-Commutativity)}: his plot demonstrates Proposition~3.2, the most mathematically subtle
result in the paper. The question is: \textbf{does it matter whether one coarse-grains first and then computes entropy, or computes entropy first and then averages spatially?} The answer is yes: these two operations do not commute for any nonlinear
functional such as the entropy density $\eta(\varrho) = \kB\varrho\ln\varrho$. The Jensen correction is
$
    \Delta_J
    \;=\;
    \eta(\bar\varrho_i)
    \;-\;
    \frac{1}{|V_i|}\int_{V_i}\eta\!\left(\fone(x,\cdot)\right)\mathrm{d}x
    \;\leq\; 0,
$ where the inequality follows from Jensen's inequality applied to the
concave function $\eta$ (Proposition~3.2, Eq.~(19)). The plot confirms
that $\Delta_J\leq 0$ for all three models and all cell sizes. The magnitude of the correction shrinks as cells grow larger. This is
because for large cells, the coarse-grained density $\bar\varrho_i$ is
very smooth, so the nonlinear correction matters less. For very fine
cells, the local density is nearly constant within each cell and
$\Delta_J\approx 0$ trivially. \textbf{Physical message.} The paper connects this to the
Buchert--Ostermann problem in relativistic cosmology (Remark~3.3): the
same non-commutativity of spatial averaging with nonlinear operations
appears in general relativity as the cosmological backreaction term.
Plot~7 is the flat-space, non-relativistic analogue of that effect. The
Jensen correction is the precise thermodynamic counterpart to the
cosmological averaging problem.
}
    \label{fig:my_p7}
\end{figure}

\begin{figure}[h] 
    \centering
    \includegraphics[width=1.0\textwidth]{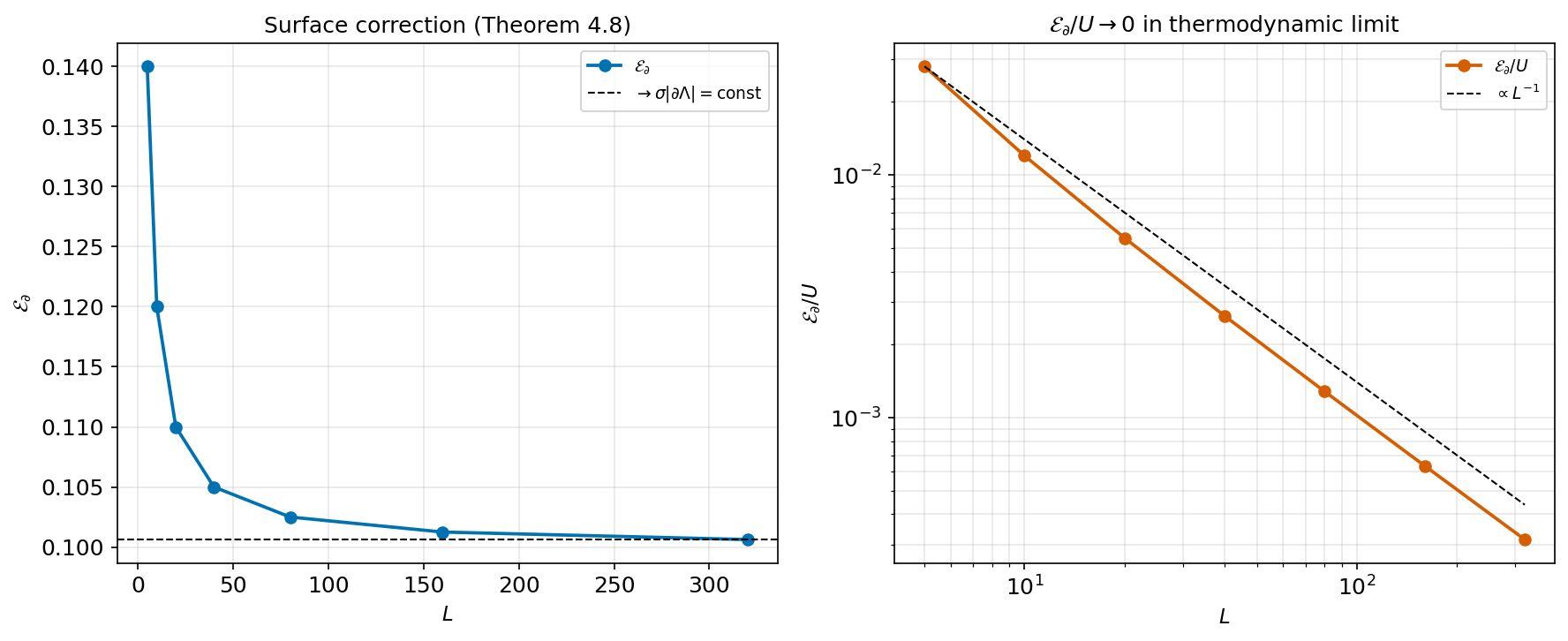}
    \caption{\textbf{The Euler Relation Surface Correction}: The standard Euler relation $U = TS - PV + \mu N$ holds exactly only for infinite systems. For any finite system, Theorem~4.8 gives the corrected relation $U \;=\; TS - PV + \mu N + \Esurf,\qquad\Esurf = O(|\partial\Lam|),$ where $\Esurf$ arises from the boundary-breaking translational invariance. (i) \textbf{Left panel( $\Esurf$ vs $L$)}: In one dimension, $|\partial\Lam| = 2$ (two boundary points), so
$\Esurf$ is essentially constant as $L$ grows. The curve quickly flattens to the asymptotic value $\sigma|\partial\Lam|$, where $\sigma$
is the surface tension. 
(ii) \textbf{Right panel: ratio $\Esurf/U$ on a log-log scale}: This is the physically important panel. The ratio falls as $L^{-1}$
(confirmed by the dashed reference line with slope $-1$), which means the surface correction becomes a smaller and smaller fraction of the
total energy as the system grows. In the thermodynamic limit, $\Esurf/U\to 0$ and the standard Euler relation is recovered. \textbf{Physical message.} This plot also makes a point emphasised in Section~5 of the paper: the surface correction
$O(|\partial\Lam|)$ and the bulk correlation correction $O(|\Lam|\,e^{-\ell/\xi})$ from Plot~4 are \emph{two entirely
different effects} with different physical origins and different functional dependences on system size. The surface correction arises
from boundary geometry; the bulk correction arises from inter-cell correlations. Standard textbooks typically conflate both under a single
$O(|\partial\Lam|)$ label, which the paper argues is imprecise. Plot~8 isolates the surface effect alone by working in the regime
$\ell\gg\xi$ where the bulk correlation correction is exponentially suppressed.}
 \label{fig:my_p8}
\end{figure}
\pagebreak
\section{The Overall Narrative}

These plots together tell a single coherent story.

\begin{enumerate}[leftmargin=1.5em, label=\arabic*.]
    \item The coarse-grained framework is \textbf{consistent with standard Gibbs statistics} in the fine-graining limit, with a controlled
    $O(M^{-2})$ convergence rate (Plot~3).

    \item The mesoscopic probabilities $\pia$ (Plot~2) are the natural objects for a reduced description, inheriting the structure of
    $\fone(x,p)$ (Plot~1) while discarding sub-cell fluctuations.

    \item Entropy is additive---and thermodynamics is extensive---if and only if inter-cell mutual information decays fast enough (Plots~4--5).
    Short-range interactions guarantee exponential decay; long-range interactions produce algebraic decay and permanent non-additivity.

    \item The failure of extensivity for long-range interactions is not merely a correction: it is a qualitatively different regime visible
    as a rising $S/N$ curve in Plot~6.

    \item Spatial averaging does not commute with nonlinear thermodynamic functionals (Plot~7). This is a real, quantifiable effect connected
    to the cosmological backreaction problem.

    \item Surface corrections to the Euler relation vanish relative to
    bulk quantities in the thermodynamic limit (Plot~8), and they are
    physically distinct from bulk correlation corrections.
\end{enumerate}

\noindent The overall conclusion is that \textbf{extensivity is not a postulate}: it is a derived consequence of microscopic stability,
temperedness, and the exponential decay of correlations. When any of these conditions fails, the entropy is non-additive, and the degree of
non-additivity is quantified by the inter-cell mutual information---a bridge between classical thermodynamics and information
theory made explicit by the coarse-graining framework of this paper.

\end{document}